%% file: main.tex
\input{macros}
\usepackage{cleveref}
\begin{document}

\title{$\BQP \subseteq \IP$ Does Not Relativize}

\date{}

\clearpage

\author[1]{Adam Bouland\thanks{ \texttt{abouland@stanford.edu}}}
\author[2]{Andrew Huang\thanks{ \texttt{a\_huang@berkeley.edu}}}
\author[3]{Anand Natarajan\thanks{ \texttt{anandn@mit.edu}}}
\author[1]{Itay Shalit\thanks{ \texttt{ishalit@stanford.edu}}}
\author[2]{Avishay Tal\thanks{ \texttt{atal@berkeley.edu}}}

\affil[1]{Department of Computer Science, Stanford University}
\affil[2]{Department of Electrical Engineering and Computer Science, University of California, Berkeley}
\affil[3]{Department of Computer Science, Massachusetts Institute of Technology}

\maketitle
\thispagestyle{empty}

\begin{abstract}
We construct an oracle relative to which $\BQP \not\subseteq \IP$, resolving a long-standing open question in quantum complexity theory. Together with recent work due to Aaronson et al., our work also gives the first oracle separation between $\IP$ and $\MIP$, answering a question dating back to Fortnow's thesis. Our separation is based on the Forrelation problem, where given Boolean functions $f$ and $g$, the goal is to determine if $f$ is correlated with the Fourier spectrum of $g$. While this task is solvable by a query-efficient quantum algorithm, we show that it admits no classical interactive protocol with polynomial communication and a polynomial-query verifier.

Our proof is based on (i) a new structural result showing how to approximate Avg-Max circuits (which are well-known to capture the power of interactive proofs in the oracular setting) by convex functions with small first and second derivatives and (ii) a novel analysis establishing that the Forrelation distribution suggested by Aaronson and Ambainis fools such functions.

Our results imply that any prover-efficient classical interactive protocol for $\BQP$ must rely on non-relativizing techniques. This might serve as a partial explanation for the lack of progress towards doubly-efficient, unconditionally sound classical verification of quantum computation.
\end{abstract}
\thispagestyle{empty}

\newpage\thispagestyle{empty}\tableofcontents
\newpage
\setcounter{page}{1}
\section{Introduction}\label{sec:intro}
\input{introduction}

\section{Technical Overview}\label{sec:overview}
\input{overview}

\section{Preliminaries}\label{sec:prelims}
\input{prelims}

\section{The Forrelation Problem}\label{sec:forrelation}
\input{forrelation}

\section{Oracle Separation}\label{sec:separation}
\input{separation}

\section{Pseudorandom Generators Against Avg-Max Circuits}\label{sec:prg}
\input{prg}

\section{AI Disclosure}
The proof idea underlying an initial version of the main theorem in this paper was generated using ChatGPT 6 Astra, giving a one-sided indistinguishability argument for a more complicated pair of distributions. After several rounds of interaction with GPT-6 and internal discussions, the authors found a two-sided bound for the original Forrelation distribution of Aaronson and Ambainis. The authors subsequently verified, simplified, and developed the argument presented here, and take full responsibility for its correctness and exposition.

\section{Acknowledgments}
We thank Scott Aaronson and \'Agi Vill\'anyi for helpful comments.
A.N. was supported by a Sloan Fellowship and NSF CAREER award 2339948.
A.T. was supported by NSF CAREER award CCF-2145474.
I.S. was supported by a Shoucheng Zhang Graduate Fellowship. 
A.B. was partially supported by the National Science Foundation under award No. 2016245 and 2440805, by the Air Force Office of Scientific Research under grant agreement FA9550-21-1-0392 and FA9550-24-1-0089, and by the DOE Office of Science under grant agreement DE-SC0025934. This work was initiated in part while the author(s) were visiting the Simons Institute for the Theory of Computing.

\bibliographystyle{alpha}
\bibliography{refs.bib}

\appendix
\input{bqp_algorithm}\label{sec:bqp_alg}
\input{general_smoothing}\label{sec:smooth-extensions}
\input{diagonalization}\label{sec:diag}
\input{bbbv_cor}\label{sec:bbbv_cor}
\end{document}

%% file: macros.tex
\documentclass[11pt]{article}
\usepackage[normalem]{ulem}
\usepackage[utf8]{inputenc}
\usepackage[skins,xparse,breakable]{tcolorbox}
\usepackage{xcolor}
\definecolor{linkblue}{HTML}{001487}
\usepackage{graphicx} 
\usepackage{stmaryrd}
\usepackage{amsmath, amssymb}
\usepackage{mdframed}
\usepackage{mathtools}
\usepackage{braket}

\usepackage{graphicx}

\usepackage[colorlinks=true,linkcolor=blue,citecolor=blue,urlcolor=blue]{hyperref}
\usepackage{authblk}
\usepackage{hyperref}
\usepackage{amsfonts}
\usepackage{braket}
\usepackage{url}
\usepackage{complexity}
\usepackage[title]{appendix}
\usepackage[all]{xy}
\usepackage[margin=1.in]{geometry}
\usepackage{physics}
\usepackage{cancel}
\usepackage{mdframed}
\usepackage{enumitem}
\usepackage{xcolor}
\usepackage{comment}
\usepackage{array}
\usepackage{stackengine,scalerel}
\usepackage{amsmath,amssymb,amsthm,mathtools}
\usepackage{booktabs}
\usepackage{algorithm}
\usepackage{algpseudocode}
\allowdisplaybreaks

\usepackage{xcolor}
\definecolor{darkgreen}{RGB}{0, 100, 0}

\theoremstyle{plain}
\newtheorem{theorem}{Theorem}[section]

\newtheorem{lemma}[theorem]{Lemma}

\newtheorem*{conjecture*}{\textbf{Conjecture}}
\newtheorem{corollary}[theorem]{Corollary}

\newtheorem{claim}{Claim}[theorem]
\newtheorem{fact}{Fact}[theorem]

\theoremstyle{definition}
\newtheorem{definition}{Definition}

\newtheorem{remark}[theorem]{Remark}
\newtheorem*{remark*}{\textbf{Remark}}

\newcommand{\eps}{\varepsilon}

\newcommand{\sgn}{\mathsf{sgn}}
\newcommand{\forr}{\mathsf{forr}}
\newcommand{\Var}{\mathsf{Var}}

\newcommand{\calA}{\mathcal A}
\newcommand{\calB}{\mathcal B}

\newcommand{\calD}{\mathcal D}

\newcommand{\calL}{\mathcal L}

\newcommand{\calN}{\mathcal N}
\newcommand{\calO}{\mathcal O}

\newcommand{\calU}{\mathcal U}

\newcommand{\sfT}{\mathsf T}

\newcommand{\bbE}{\mathbb E}

\newcommand{\bbN}{\mathbb N}

\newcommand{\bbR}{\mathbb R}

\renewcommand{\E}{\mathop{\mathbb E\/}}
\renewcommand{\R}{{\mathbb R}}

\newcommand{\ACzero}{{\mathsf{AC}^0}}

%% file: introduction.tex
One of the central theoretical and practical questions in quantum complexity is whether there exist doubly-efficient interactive proofs for quantum computation.\footnote{We point the reader to a recent survey by Aaronson \cite{Aar21} on open problems in quantum query complexity.} Namely, can an efficient quantum prover convince a classical verifier of the correctness of a quantum computation?

In the unrelativized setting, it is well-known that $\BQP \subseteq \IP$ due to the celebrated $\IP = \PSPACE$ theorem of Shamir \cite{Sha92}, but the prover in the aforementioned protocol does not run efficiently even when restricted to proving $\BQP$ computations. Indeed, although the work of Mahadev \cite{Mah18} demonstrated that doubly-efficient classical verification of quantum computation is possible under cryptographic assumptions, it has remained open for several decades whether an information-theoretic analogue exists.

Progress towards this question mostly comes from a relatively small number of works. First, an early result due to McKague \cite{McK12} showed that Recursive Fourier Sampling, a $\BQP$ problem which is suspected to lie outside of $\PH$, has a doubly-efficient interactive proof. Second, the work of Reichardt, Unger, and Vazirani \cite{RUV13}, along with countless follow-ups, has demonstrated the existence of doubly-efficient \emph{entangled multi-prover} interactive proof systems for $\BQP$. Recent work due to Aaronson, Natarajan, Tal, and Vill\'anyi \cite{ANTV26} has shown that there is a relativizing \emph{classical multi-prover} interactive proof system for $\BQP$, but it seems difficult to make the provers in their protocol efficient (and indeed, they appear to require at least the power of $\PostBQP = \PP$).

In this work, we give some evidence as to why progress has stalled thus far: there is no \emph{relativizing} prover-efficient interactive proof for $\BQP$. In fact, we prove something stronger---there is no relativizing interactive proof for $\BQP$, even ignoring prover efficiency! 
\begin{theorem}[Informal]\label{theorem:oracle_separation}
    There exists a classical oracle $\calO$ relative to which $\BQP^{\calO} \not\subseteq \IP^{\calO}$.
\end{theorem}

This implies any interactive proof for $\BQP$ must contain a non-relativizing ingredient---such as a cryptographic assumption, or an initial complete problem which is $\BQP$-complete but not $\BQP^\calO$-complete under $\P^\calO$ reductions for all $\calO$. In fact, Mahadev's protocol has both of these non-relativizing ingredients, as both the cryptographic hardness of LWE and the $\BQP$-hardness of Local $XZ$-Hamiltonians are non-relativizing (see \cite{ANTV26} for a more detailed discussion).

As a corollary of Theorem \ref{theorem:oracle_separation} and the recent relativizing proof that $\BQP \subseteq \MIP$ \cite{ANTV26}, we derive several new oracle separations which were previously unknown, including the curious case of $\IP$ versus $\MIP$ (answering an open question of Fortnow \cite{For89}):\footnote{Section~\ref{sec:overview} explains explicitly why our lower bound techniques differentiate $\MIP$ and $\IP$.}
\begin{corollary}[Informal]\label{cor:more_separations}
    Relative to the oracle $\calO$ in Theorem \ref{theorem:oracle_separation}, there exists a language $\calL$ such that $\calL \in \BQP^{\calO} \subseteq \MIP^{\calO} \cap \QIP^{\calO}$ but $\calL \notin \IP^{\calO}$.
    
    As a consequence, relative to $\calO$, $\IP_{\BQP}^{\calO} \neq \BQP^{\calO}$, $\IP^{\calO} \neq \QIP^{\calO}$, and $\IP^{\calO} \neq \MIP^{\calO}$.
\end{corollary}

Our results follow by showing that the Forrelation problem due to Aaronson and Ambainis \cite{Aar10,AA15} does not have an interactive proof. To prove such lower bounds, we combine the expectation/average-maximization characterization of $\IP$ circuits due to Aiello, Goldwasser, and H{\aa}stad \cite{AGH90} with a proof that Forrelation fools such Avg-Max circuits. While our lower bound techniques bear some resemblance to the oracle separation of $\BQP$ from $\PH$ due to Raz and Tal \cite{RT22}, they are not a mere strengthening of the Fourier growth bounds of \cite{RT22}.

By combining the results of Theorem \ref{theorem:oracle_separation} with the classic hybrid argument of \cite{BBBV97}, we also obtain the following oracle separation.
\begin{corollary}\label{cor:simons_like}
    There exists a classical oracle $\calO$ relative to which $\BQP^{\calO} \cap \NP^{\calO} \not\subseteq \IP_{\BQP}^{\calO}$.
\end{corollary}
We give a proof of Corollary \ref{cor:simons_like} (which is relatively straightforward) in Appendix \ref{sec:bbbv_cor}.

%% file: overview.tex
We construct a classical oracle $\calO$ such that $\BQP^{\calO} \not\subseteq \IP^{\calO}$ by showing that Forrelation is easy for quantum algorithms but hard for interactive proofs. 

Following prior work \cite{AGH90,Aar10,RT22}, separating $\BQP$ from $\IP$ boils down to a pseudorandomness question: exhibit a distribution over exponentially long Boolean strings that is on one hand easy to distinguish from uniform by a polynomial-time quantum algorithm (making queries to the string) and on the other hand, indistinguishable from uniform by any Avg-Max circuits of polynomial depth equipped with decision trees of polynomial depth on the leaves.

\paragraph{The Forrelation distribution.}
Following \cite{Aar10,AA15}, we consider the following  distribution,\footnote{This is the precise distribution  suggested by Aaronson \cite{Aar10}, with no modifications.} denoted $\calD_n$, over $\{-1, 1\}^{2N}$, where $N = 2^n$.
\begin{quote}
    Sample $N$ independent standard Gaussians $x_1, \ldots, x_N \sim \calN(0, 1)$ before outputting $(\sgn(x), \sgn(H_N \cdot x))$. 
    
    Equivalently, sample a vector $z$ from the centered Gaussian $D$ with a covariance matrix
        \[
            \Sigma = \begin{pmatrix} I_N & H_N \\ H_N & I_N \end{pmatrix},
        \]
    before outputting $\sgn(z)$.
\end{quote}
Denote by $\calU_n$ the uniform distribution over $\{\pm 1\}^{2N}$.
The $\BQP$ algorithm for Forrelation due to Aaronson and Ambainis~\cite{AA15} makes $1$ query and runs in $\poly(n)$ time, and distinguishes between $\calD_n$ and $\calU_n$ with constant advantage.
A straightforward amplification of it distinguishes between $\calD_n$ and $\calU_n$ with high probability. Our goal therefore is to argue that these two distributions are indistinguishable to $\IP$ verifiers.

\paragraph{An IP lower bound.}
Let $\langle P, V \rangle$ be an interactive proof system that attempts to distinguish between $\calD_n$ and $\calU_n$. Consider the acceptance function
    \[ a_V(f, g)=\max_P\Pr[\langle V^{f,g},P\rangle=1] \]
associated with $V$. Our main technical result shows that for every $\IP$ $\langle P, V \rangle$,
\[
    \left|
      \E_{\mathcal D_n}[a_V]
       -\E_{\mathcal U_n}[a_V]
    \right|
       \leq 2^{-\Omega(n)}.
\]

\paragraph{Proof outline.} 
Before we go into further detail, we provide a rough outline of the proof, consisting of six steps, for the reader to follow along. The novelty of the proof is centered in steps 2 and 3; the other steps are either standard in the literature, or employ similar techniques to those used in the oracle separation between $\BQP$ and $\PH$~\cite{RT22,Wu22}. 

\begin{enumerate}
    \item \textbf{Representing IP Protocols with Oracle Queries as Avg-Max Circuits.}

    By the standard transformation of \cite{AGH90}, we can express the optimal acceptance probability of the verifier using a circuit with average and maximum gates with affine leaves, obtaining a convex extension $C:\mathbb R^{2N}\to \mathbb R$ (see Lemma \ref{lemma:am_characterization}) of $a_V:\{-1,1\}^{2N} \to [0,1]$.

    \item \textbf{New Structural Result for Avg-Max Circuits.}
    To prove lower bounds against $\IP$ verifiers, one generally needs to isolate relevant structural properties of their acceptance functions. For example, in the oracle separation between $\BQP$ and $\PH$, the key structural property is a poly-logarithmic bound on the level-two Fourier mass of the multilinear extension of the Boolean circuit representing the acceptance function of a $\PH$ algorithm. This yields control of the extension's second derivatives, which enables the separation. 

    However, we are not aware of an analogous poly-logarithmic bound on the level-two Fourier mass of the function $C$ defined above. Instead, we replace the maximum gates in our convex extension by smooth approximations, obtaining a smooth \emph{and} convex function $F$ that approximates the verifier’s optimal acceptance probability on Boolean inputs and has controlled $\ell_1$-gradient and Hessian norms (i.e., first and second derivatives). We emphasize that the convexity of our function will be crucial to our lower bound.\footnote{As we show in Appendix~\ref{sec:smooth-extensions}, every  function from $\{-1,1\}^{2N}$ to $[-1,1]$ can be approximated by a smooth extension with small derivative norms. Thus, without convexity the structural result is vacuous.}

    \item \textbf{From Signs to Approximate Signs.}
    We compare the expectations of $F$ under Gaussian inputs with varying amounts of independent noise. That is, we compare $F(\sgn(Z))$ and $F(\sgn(\sqrt{t} Z + \sqrt{1-t} D))$, where $Z$ is an independent multivariate Gaussian and $D$ is the above Forrelated Gaussian distribution. (The only property we rely on with respect to $D$ is that it is a centered multivariate Gaussian with variances $1$ and small covariances, $|\Sigma_{i,j}|\le 1/\sqrt{N}$  for $i\neq j$.)
    Via careful analysis, the expectation of a \emph{smooth} function on $Z$ is similar to that on $\sqrt{t} Z + \sqrt{1-t} D$.
    However, the composition of $F$ and $\sgn$ is not smooth. We therefore replace the sign function by a smooth approximation function $h$. Convexity and $\ell_1$-gradient bounds are crucial for controlling the error introduced by this replacement (see Lemma \ref{lemma:smoothening}).

    \item \textbf{Comparing the Expectations Under Approximate Signs.}
    By varying $t$ continuously, we bound the expression $|\bbE[F(h(Z))]-\bbE[F(h(\sqrt{1/2}Z + \sqrt{1/2}D))]|$ (see Lemma \ref{lemma:interpolation}). Combining this with the rounding bound from Step 3 gives a comparison of expectations on signs.

    \item \textbf{Removing the Independent Gaussian Contribution.}
    To move to our final distributions, we show that the bound given in the previous steps can be extended to the setting where a fixed vector is added to both Gaussian inputs before taking their signs (see Theorem \ref{theorem:comparison}). Using a hybrid argument then allows us to compare the distributions $\calD_n$ and $\calU_n$ (see Corollary \ref{cor:recursive}).

    \item \textbf{Obtaining the Oracle Separation.}
    Finally, we choose the appropriate parameters to make the verifier's expected optimal acceptance probabilities under $\calD_n$ and $\calU_n$ differ by at most $2^{-\Omega(n)}$ (see Corollary \ref{cor:values}). A diagonalization argument then yields the oracle separation (see Corollary \ref{cor:separation}).
\end{enumerate}

\noindent We are now ready to explain each of the steps of our proof in more detail.

\paragraph{Step 1: Representing IP Protocols with Oracle Queries as Avg-Max Circuits.} 
Fix an interactive proof system $\langle P, V \rangle$. By applying the relativizing transformation of Goldwasser and Sipser \cite{GS86}, we may assume without loss of generality that the verifier's messages consist of uniformly random coins, and that it makes all of its oracle queries after the interaction ends. 

Following the work of Aiello, Goldwasser, and H\r{a}stad \cite{AGH90}, $a_V$ can be represented by a tree of \textit{average} and \textit{maximum} gates: average gates correspond to the verifier's random messages, and maximum gates to the prover's responses. We also need to represent the verifier's (potentially adaptive) oracle queries. Let $L$ bound each message's length and the number of oracle queries. After fixing a completed transcript, the adaptive verifier queries are represented by a binary decision tree of depth at most $L$. For a path $\gamma$ that queries locations $i_1, \ldots, i_k$ (where $k \leq L$), receives answers $a_1, \ldots, a_k \in \{-1, 1\}$, and outputs a decision bit $b_\gamma \in \{0, 1\}$, define
\[
    \ell_\gamma(x) = b_\gamma - \sum_{j = 1}^k \frac{1-a_jx_{i_j}}{2}.
\]
Here, $x\in \{-1,1\}^{2N}$. Any input $x$ that is consistent with $a_1,\ldots,a_k$ on $i_1,\ldots, i_k$ has value equal to the verifier's output $b_\gamma$, while every inconsistent path has value at most zero. Hence $\max_\gamma\ell_\gamma(x) = b_\gamma$. Additionally, $\ell_\gamma$ has coefficient $\ell_1$-norm at most $L$.

Consequently, $a_V$ has an extension $C:\mathbb R^{2N}\to\mathbb R$ computed by a tree of average and maximum gates, with fan-in at most $2^L$ and affine leaves of coefficient $\ell_1$-norm at most $L$. For a polynomial-time verifier, $d, L = \poly(n)$, where $d$ is the tree's depth. Both gate types preserve convexity, so $C$ is convex.

\paragraph{Step 2: New Structural Result for Avg-Max Circuits.}
To use Gaussian interpolation, we need a differentiable extension. Observing that the average gates are already smooth, we replace every maximum gate by the softmax gate
\[
    S_\lambda(u_1,\ldots,u_k)
       =\frac{1}{\lambda}
          \log\!\left(\sum_{j=1}^k e^{\lambda u_j}\right).
\]
It can be shown that the resulting function $F$ remains convex and satisfies the following bounds over its entire domain:
\[
    0\leq F-C\leq\frac{dL\log 2}{\lambda},
    \qquad
    \sum_{i \in [2N]} |\partial_i F(x)| \leq L,
    \qquad
    \sum_{i, j \in [2N]}|\partial_{ij} F(x)|\leq 2\lambda dL^2.
\]

\paragraph{Step 3: From Signs to Approximate Signs.}
Write $m = 2N$. The distribution $\calD_n$ consists of the signs of a centered Gaussian vector $D \sim \calN(0, \Sigma)$. On the other hand, the uniform distribution $\calU_n$ can be thought of as the signs of a Gaussian vector $Z \sim \calN(0, I_m)$. It therefore  suffices to compare $\bbE[F(\sgn(D))]$ and $\bbE[F(\sgn(Z))]$.

The Hessian bound derived above controls how the expectation of a smooth function on an input sampled from a Gaussian distribution changes when Gaussian covariances change. However, it cannot be applied directly here, because $F \circ \sgn$ is discontinuous. To remedy this issue, we first replace the sign function by a smooth approximation and prove that this changes the expectation only slightly by making crucial use of convexity. For $\frac{1}{2} \leq t \leq 1$, consider the random variable
    \[ G_t = \sqrt{t}\,Z + \sqrt{1-t}\,D. \]
The independent Gaussian contribution $\sqrt{t}\,Z$ has variance at least $1/2$ in every coordinate. For $0 < \delta \leq 1/2$, define the scalar function
\[
    h_\delta(u)=\mathbb E_{W\sim\mathcal N(0,1)}
       \left[\sgn
       \left(\sqrt{1-\delta^2}\,u+\delta W\right)\right],
\]
and apply it coordinatewise. This function takes values in $[-1,1]$,
has derivative bounded by $1/\delta$, and approximates the sign
function except near zero (see Figure~\ref{fig:smoothed-sign}).
\begin{figure}[ht]
\begin{mdframed}
\caption{\textbf{A Gaussian smoothing of the sign function.}
The function $h_\delta(x)=2\Phi(\sqrt{1-\delta^2}\,x/\delta)-1$ transitions smoothly from $-1$ to $1$ near the origin.}
\label{fig:smoothed-sign}
\begin{center}
\begin{tikzpicture}[
    x=1.25cm, y=1.8cm,
    >=stealth,
    every node/.style={font=\small}
]

\fill[blue!8] (-1,-1.12) rectangle (1,1.12);
\draw[blue!45,dashed] (-1,-1.12) -- (-1,1.12);
\draw[blue!45,dashed] ( 1,-1.12) -- ( 1,1.12);

\draw[->] (-4.5,0) -- (4.6,0) node[right] {$x$};
\draw[->] (0,-1.25) -- (0,1.43) node[above] {$y$};
\foreach \x/\lab in {-1/{-2\delta},1/{2\delta}} {
    \draw (\x,-0.035) -- (\x,0.035);
    \node[below=3pt,fill=blue!8,inner sep=1pt] at (\x,0) {$\lab$};
}

\draw[black!48,densely dashed,very thick] (-4.25,-1) -- (0,-1);
\draw[black!48,densely dashed,very thick] (0,1) -- (4.25,1);
\draw[black!48,fill=white] (0,-1) circle (1.5pt);
\draw[black!48,fill=white] (0, 1) circle (1.5pt);
\node[black!65,above=5pt] at (3.15,1) {$\sgn(x)$};

\draw[red!75!black, thick,line cap=round,line join=round]
    plot[smooth,tension=0.5] coordinates {
        (-4.25000,-1.0000000) (-4.18750,-1.0000000) (-4.12500,-1.0000000)
        (-4.06250,-1.0000000) (-4.00000,-1.0000000) (-3.93750,-1.0000000)
        (-3.87500,-1.0000000) (-3.81250,-1.0000000) (-3.75000,-1.0000000)
        (-3.68750,-1.0000000) (-3.62500,-1.0000000) (-3.56250,-1.0000000)
        (-3.50000,-1.0000000) (-3.43750,-1.0000000) (-3.37500,-1.0000000)
        (-3.31250,-1.0000000) (-3.25000,-1.0000000) (-3.18750,-1.0000000)
        (-3.12500,-1.0000000) (-3.06250,-1.0000000) (-3.00000,-1.0000000)
        (-2.93750,-1.0000000) (-2.87500,-1.0000000) (-2.81250,-1.0000000)
        (-2.75000,-1.0000000) (-2.68750,-0.9999999) (-2.62500,-0.9999998)
        (-2.56250,-0.9999996) (-2.50000,-0.9999993) (-2.43750,-0.9999987)
        (-2.37500,-0.9999976) (-2.31250,-0.9999955) (-2.25000,-0.9999920)
        (-2.18750,-0.9999858) (-2.12500,-0.9999752) (-2.06250,-0.9999574)
        (-2.00000,-0.9999277) (-1.93750,-0.9998793) (-1.87500,-0.9998012)
        (-1.81250,-0.9996776) (-1.75000,-0.9994845) (-1.68750,-0.9991876)
        (-1.62500,-0.9987381) (-1.56250,-0.9980680) (-1.50000,-0.9970841)
        (-1.43750,-0.9956616) (-1.37500,-0.9936364) (-1.31250,-0.9907968)
        (-1.25000,-0.9868762) (-1.18750,-0.9815456) (-1.12500,-0.9744083)
        (-1.06250,-0.9649978) (-1.00000,-0.9527791) (-0.93750,-0.9371560)
        (-0.87500,-0.9174847) (-0.81250,-0.8930935) (-0.75000,-0.8633111)
        (-0.68750,-0.8275003) (-0.62500,-0.7850971) (-0.56250,-0.7356533)
        (-0.50000,-0.6788789) (-0.43750,-0.6146807) (-0.37500,-0.5431946)
        (-0.31250,-0.4648067) (-0.25000,-0.3801609) (-0.18750,-0.2901509)
        (-0.12500,-0.1958959) (-0.06250,-0.0987002) (0.00000,0.0000000)
        (0.06250,0.0987002) (0.12500,0.1958959) (0.18750,0.2901509)
        (0.25000,0.3801609) (0.31250,0.4648067) (0.37500,0.5431946)
        (0.43750,0.6146807) (0.50000,0.6788789) (0.56250,0.7356533)
        (0.62500,0.7850971) (0.68750,0.8275003) (0.75000,0.8633111)
        (0.81250,0.8930935) (0.87500,0.9174847) (0.93750,0.9371560)
        (1.00000,0.9527791) (1.06250,0.9649978) (1.12500,0.9744083)
        (1.18750,0.9815456) (1.25000,0.9868762) (1.31250,0.9907968)
        (1.37500,0.9936364) (1.43750,0.9956616) (1.50000,0.9970841)
        (1.56250,0.9980680) (1.62500,0.9987381) (1.68750,0.9991876)
        (1.75000,0.9994845) (1.81250,0.9996776) (1.87500,0.9998012)
        (1.93750,0.9998793) (2.00000,0.9999277) (2.06250,0.9999574)
        (2.12500,0.9999752) (2.18750,0.9999858) (2.25000,0.9999920)
        (2.31250,0.9999955) (2.37500,0.9999976) (2.43750,0.9999987)
        (2.50000,0.9999993) (2.56250,0.9999996) (2.62500,0.9999998)
        (2.68750,0.9999999) (2.75000,1.0000000) (2.81250,1.0000000)
        (2.87500,1.0000000) (2.93750,1.0000000) (3.00000,1.0000000)
        (3.06250,1.0000000) (3.12500,1.0000000) (3.18750,1.0000000)
        (3.25000,1.0000000) (3.31250,1.0000000) (3.37500,1.0000000)
        (3.43750,1.0000000) (3.50000,1.0000000) (3.56250,1.0000000)
        (3.62500,1.0000000) (3.68750,1.0000000) (3.75000,1.0000000)
        (3.81250,1.0000000) (3.87500,1.0000000) (3.93750,1.0000000)
        (4.00000,1.0000000) (4.06250,1.0000000) (4.12500,1.0000000)
        (4.18750,1.0000000) (4.25000,1.0000000)
    };
\node[red!75!black,below=7pt] at (3.1,0.94) {$h_\delta(x)$};

\draw[blue!55] (-1,-1.18) -- (-1,-1.49);
\draw[blue!55] ( 1,-1.18) -- ( 1,-1.49);
\draw[<->,blue!65!black,thick] (-1,-1.40) -- (1,-1.40)
    node[midway,below=4pt] {$4\delta$};
\end{tikzpicture}
\end{center}
\end{mdframed}
\end{figure}
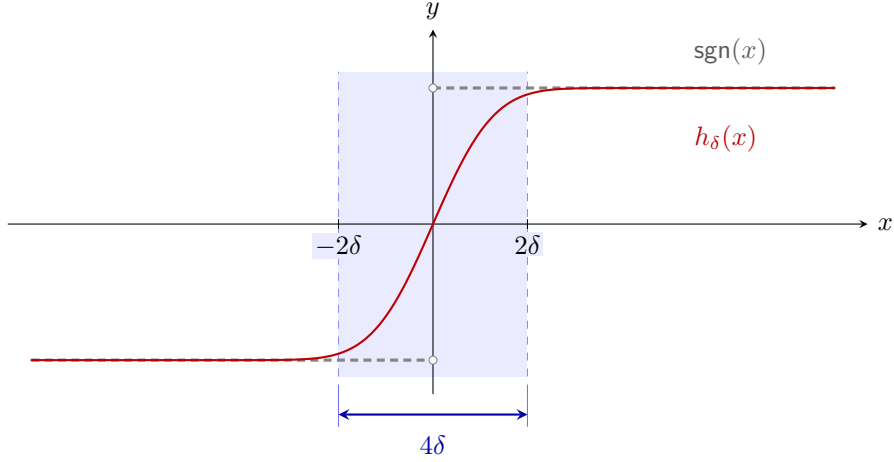
Our rounding lemma proves that
\begin{equation}
\label{eq:rounding_bound}
    0\leq
\E[ F(\sgn(G_t))]-
   \E[F(h_\delta(G_t))]
    \leq O\!\left(L\delta\sqrt{\log(2m)}\right).
\end{equation}
The logarithmic dependence on $m$ is essential: a bound proportional
to the number of coordinates would be too large for our application.

The difficulty is that some coordinates may lie near zero and have large rounding errors. Convexity lets us control their contribution: writing $S = \sgn(G_t)$, we have
    \begin{equation}\label{eq:convexity}
 F(S)-F(h_\delta(G_t)) \leq \sum_{i = 1}^m \partial_i F(S) \bigl(S_i-h_\delta((G_t)_i)\bigr). \end{equation}
    The right hand side is the inner product of $\nabla F(S)$ and $(S-h_{\delta}(G_t))$,  where we know that $\|\nabla{F}(S)\|_1\le L$.
    It is tempting to try bounding the inner product by $L\cdot \|(S-h_{\delta}(G_t))\|_{\infty}$. However, typically $\|(S-h_{\delta}(G_t))\|_{\infty}$ is large as we have $m\gg 1/\delta$ coordinates and many would fall in the interval of length $O(\delta)$ where $h_{\delta}(\cdot)$ does not approximate the sign well. The key insight is that even after we condition on $S=s$ (and $D=d$), $G_t$ still remains somewhat random allowing us to average over $(S - h_{\delta}(G_t))|S=s$, and this average is typically small in \emph{all coordinates}.
    
We provide more details. Conditioning on $S = s$ and $D = d$ fixes the coefficients $\partial_i F(S)$ in Eq.~\eqref{eq:convexity}, whose total absolute value is at most $L$. But note that once $D = d$ is fixed, the coordinates of $G_t$ are independent Gaussians! Conditioning further on $S = s$ therefore restricts each coordinate to the positive or negative half-line specified by its sign, with no additional restriction from the other coordinates’ signs. Since smoothing errors are concentrated near zero, a bound on this conditional Gaussian density, applied in particular near zero, gives an expected absolute error of $O(\delta(1+|d_i|))$ in coordinate $i$. We conclude that the conditional expectation of the sum is at most $O(L\delta(1+\|d\|_\infty))$. Averaging over the values of $D$ and $S$ and using a Gaussian tail bound $\bbE\|D\|_\infty = O(\sqrt{\log(2m)})$ then yields the upper bound in \eqref{eq:rounding_bound}.

For the lower bound, let $G_t'$ be an independent copy of $G_t$. The vector $\sqrt{1-\delta^2}\,G_t+\delta G_t'$ has the same distribution as $G_t$, and its sign vector has conditional mean $h_\delta(G_t)$ given $G_t$. The lower bound then follows from Jensen's inequality.

\paragraph{Step 4: Comparing the Expectations Under Approximate Signs.}
We interpolate from $G_{1/2} = (Z+D)/\sqrt{2}$ to $G_1 = Z$. Set $Q = F \circ h_\delta$, write $\beta = \max_{i \neq j} |\Sigma_{ij}|$, and let $K$ satisfy $\sum_{i,j}|\partial_{ij}F(x)| \leq K$ on $[-1, 1]^m$. Gaussian interpolation gives
    \[ \frac{d}{dt}\bbE[Q(G_t)] = -\frac{1}{2}\sum_{i\ne j}\Sigma_{ij}\, \bbE\bigl[\partial_{ij}Q(G_t)\bigr]. \]
There are no diagonal terms because the coordinate variances remain fixed. For $i \neq j$, the chain rule gives
\[
\partial_{ij}Q(x) = (\partial_{ij}F)(h_\delta(x))\, h_\delta'(x_i)h_\delta'(x_j).
\]
Using $|h_\delta'| \leq 1/\delta$, we therefore obtain
\[
\left|\frac{d}{dt}\bbE[Q(G_t)]\right| \leq \frac{\beta K}{2\delta^2}.
\]
Integrating over $t \in [1/2,1]$ and applying the rounding bound \eqref{eq:rounding_bound} at both endpoints shows that the expectations of $F$ on $\sgn((Z+D)/\sqrt{2})$ and $\sgn(Z)$ differ by at most $O(\varepsilon_0)$, where
\[
\varepsilon_0 = L\delta\sqrt{\log(2m)}+\frac{\beta K}{\delta^2}.
\]
The first term accounts for smoothing the signs, and the second for
changing the covariances.

\paragraph{Step 5: Removing the Independent Gaussian Contribution.}
We bounded the difference between $\bbE[F(\sgn(G_t))]$ at $t = 1$ and $t = \frac{1}{2}$, since the rounding bound requires $t \geq \frac{1}{2}$. However, our goal is to compare $G_1 = Z$ with $G_0 = D$. To this end, we use a hybrid argument that halves the independent variance in successive steps:
\[
H_j = 2^{-j/2}Z + \sqrt{1-2^{-j}}\,D, \qquad j = 0, \ldots, T.
\]
Thus $H_0 = Z$, while $H_T$ approaches $D$  for large $T$. We extend the preceding comparison to
\[
\sgn(a + Z)
\quad\text{and}\quad
\sgn\!\left(a + \frac{Z+D}{\sqrt{2}}\right),
\]
uniformly over the common shift $a \in \bbR^{2N}$. 
Comparing $\sgn(H_j)$ to $\sgn(H_{j+1})$ reduces to comparing $\sgn(a+H_0)$ to $\sgn(a+H_1)$ under a shift $a \in \bbR^{2N}$. Next, we explain how to extend the previous analysis to any arbitrary shift $a \in \bbR^{2N}$.

The key observation is that in coordinates $i$ where $|a_i|$ is large, the shift almost determines the corresponding signs: Gaussian noise is unlikely to overcome them. We can therefore fix those signs at negligible cost, reducing to the case of bounded shifts. For the remaining coordinates with bounded shifts, we use Lemma \ref{lemma:smoothening}, which extends the earlier sign-smoothing bound to inputs with a bounded shift. Combining this bound with the same interpolation argument gives an $O(\varepsilon_0)$ comparison. Thus, the error per step stays essentially unchanged as the independent noise decreases.

After $T = O(\log m)$ hybrid steps, the residual noise  is very small. Gaussian coordinates are unlikely to lie close enough to zero for this noise to change their signs, so the final comparison of $\sgn(H_T)$ with $\sgn(D)$ incurs only a negligible error. Summing the errors over the steps gives, for $L \leq m$,
\[
\left|\bbE[F(\sgn(Z))] - \bbE[F(\sgn(D))]\right| \leq O(\log m) \bigl(\varepsilon_0 + m^{-10}\bigr).
\]

\paragraph{Step 6: Obtaining the Oracle Separation.}
For Forrelation, $\beta = N^{-1/2}$. By setting our parameters appropriately, one can combine our previous results to show that
    \[ \left|\E_{\mathcal D_n}[a_V] - \E_{\mathcal U_n}[a_V]\right| \leq 2^{-\Omega(n)}. \]
To obtain our oracle separation, we can then apply a standard diagonalization argument. For more details, see Appendix \ref{sec:diag}.

\section{Discussion}
\subsection{Relation to the Raz--Tal Separation.}
Raz and Tal's analysis~\cite{RT22,Wu22} also admits a Hessian interpretation. They describe the behavior of a $\PH$ algorithm using multilinear extensions $f:\R^{2N}\to \R$ of constant-depth circuits (aka $\ACzero$ circuits). 
There, $\partial_{ij}f(0) = \widehat f(\{i,j\})$ for $i \neq j$ and $\partial_{ii}f(x) = 0$ everywhere.
Thus, level-two Fourier growth controls the Hessian at $0$. 
Tal~\cite{Tal17} obtains tight level-two Fourier bounds for constant-depth circuits, but with an exponential dependence on the depth.
Furthermore, using closure under restrictions, they translate these bounds to a bound on $\partial_{ij}f(x)$ for any point in an interior box of the Boolean cube, $x\in [-1/2,1/2]^{2N}$. Thus, the structural result for constant-depth circuits resembles that for the smoothed-Avg-Max circuits with two differences: (i) the bounds on the second derivatives for $\ACzero$ are not everywhere, and (ii) the extension from Boolean inputs to $\R^{2N}$ is multilinear in the case of $\ACzero$ whereas for smoothed Avg-Max circuits the extension is convex and non-multilinear.

For $\ACzero$,  multilinearity means that $\E[f(\sgn(Z))] = f(0)$. This allows \cite{RT22} to compare (a scaled version of) the forrelated Gaussian $D$ to the all-zero vector. 
Step~4 in the above outline is very similar to the analysis there, but with a different $h$ function (namely, $h(x)$ truncates $x$ to the interval $[-1/2,1/2]$).
Rather than proving Fourier bounds for the multilinear extension of Avg-Max circuits, we construct a smooth non-multilinear convex approximation with derivative bounds polynomial in $d,L$, allowing polynomially many rounds.

Another relation to Raz-Tal is more high level. 
Their result gives lower bounds against algorithms in the Polynomial-time Hierarchy ($\PH$) that make constant and even $o(n/\log n)$ levels of alternation. By \cite{GS86}, this means that they imply the same lower bounds for any interactive protocol making $o(n/\log n)$ rounds of communication. Again, our results extend to any $\poly(n)$ rounds, and even $\exp(n/10)$ rounds.

\subsection{Why Multiple Provers Evade the Lower Bound.}
An important feature of the $\IP$ representation is that we never maximize over a complete prover strategy at once. Since the prover sees the entire public transcript, its optimal acceptance probability can be computed by backward induction: at each prover turn we maximize only over the prover's next message (or, equivalently, over its bits sequentially). Thus the resulting AVG/MAX tree has only polynomial depth and each MAX gate has at most singly exponential fan-in.

This local maximization need not be valid for $\MIP$. A prover's response may depend only on its local transcript, so two different global transcripts that induce the same local view must receive the same response. Independently maximizing at the corresponding nodes of the global transcript tree would violate this consistency constraint and would effectively reveal to the prover information held by the other provers. A direct way to preserve consistency is therefore to maximize over a complete prover strategy table. Such a table may contain $S = 2^{\poly(n)}$ bits. Consequently, one must either use a MAX gate over $2^S$ strategies, whose softmax approximation error is $S/\lambda$, or expand this maximum into $S$ binary MAX layers. In either representation, the parameter controlling our smoothing and Hessian bounds becomes exponential, so the quantitative lower-bound argument no longer applies. This explains how the relativizing $\MIP$ protocol for $\BQP$ of~\cite{ANTV26} evades our bound.

\subsection{Pseudorandom Generators Against Avg-Max Circuits.}
Our indistinguishability result between $\calD_n$ and $\calU_n$ for average-maximum circuits holds for any other distribution $\mathcal{Q}_n$ that is defined as the sign function applied to a Gaussian with unit coordinate variances and a bound on off-diagonal covariances. To apply it beyond Forrelation, choose explicit unit vectors $v_1,\ldots,v_m$ with small pairwise inner products in a low-dimensional space. For a standard Gaussian vector $g$ in that space, the coordinates $\langle v_i,g\rangle$ have variance one and covariances $\langle v_i,v_j\rangle$. Based on our indistinguishability result, their signs therefore fool average-maximum circuits when the inner products are sufficiently small. 

To construct the vectors explicitly, we use Ta-Shma's balanced-code construction~\cite{TaShma17}, following the approach of~\cite{CHLT19}. For a prescribed bound $\beta$ on the absolute pairwise inner products, this gives unit vectors $v_1, \ldots, v_m \in \bbR^r$ with $r = \poly(\log m, 1/\beta)$.

To generate one output, draw $r$ independent standard Gaussian numbers $g_1, \ldots, g_r$, set $g = (g_1, \ldots, g_r)$, and output
    \[ \bigl(\sgn\langle v_1, g\rangle, \ldots, \sgn\langle v_m, g\rangle\bigr). \]
Thus all $m$ output signs are computed from the same $r$ Gaussian samples.

%% file: prelims.tex
For a differentiable function $F: \bbR^m \to \bbR$, we define the $\ell_1$-gradient norm of $F$ as
    \[ \|\nabla F(x)\|_1 := \sum_{i \in [m]} |\partial_i F(x)|. \]
If $F$ is twice differentiable, then we can define the Hessian norm of $F$ as
    \[ \|\nabla^2F(x)\|_1 := \sum_{i, j \in [m]}|\partial_{ij} F(x)|. \]

Let $\Phi$ and $\varphi$ denote the CDF and PDF, respectively, of a standard Gaussian with mean zero and variance 1.

\subsection{Gaussian Random Variables}
We will use the following facts about Gaussian random variables throughout the proof.

\begin{fact}[Gaussian sign identity (see \cite{GW95})]\label{fact:sign_identity}
For any two jointly Gaussian, centered random variables $X, Y$ with positive variances with Pearson correlation $\rho$,
    \[ \bbE[\sgn(X)\sgn(Y)] = \frac{2}{\pi} \arcsin(\rho). \]
\end{fact}

\begin{fact}[Expected maximum of standard Gaussians]\label{fact:exp_max}
Let $T_1, \ldots, T_m$ be, not necessarily independent, standard Gaussian random variables, i.e., $T_i \sim \mathcal N(0, 1)$. Then
\[
\mathbb E\!\left[\max_{i \in [m]}|T_i|\right] \leq 4  \sqrt{\ln 2m}.
\]
\end{fact}
\begin{proof}
Define the random variable $M := \max_i |T_i|$ and let $k$ be a positive integer. Then, $M^{2k} \le \sum_{i=1}^m |T_i|^{2k}$, and so
    \[ \E[M^{2k}] \leq \sum_{i = 1}^m \E[|T_i|^{2k}] = \sum_{i = 1}^m (2k-1)!! \leq m \cdot (2k)^k,\]
where the equality uses the standard formula for an even moment of a Gaussian. Applying Jensen's inequality, we get that
    \[ \mathbb E[M] \le \left(\mathbb E[M^{2k}]\right)^{1/(2k)} \le m^{1/(2k)} \cdot \sqrt{2k}. \]
Picking $k = \lceil{\ln 2m \rceil}$ completes the proof.
\end{proof}

\begin{lemma}[Smart path method, \cite{Tal10}]\label{lemma:smart_path}
    Let $Z$, $D$ be independent centered Gaussian vectors in $\bbR^m$ with covariance matrices $\Sigma_Z$, $\Sigma_D$, and $Q$ be any smooth function such that any partial derivative of $Q$ is bounded on $\bbR^m$. Define the random variable
        \[ G(t) := \sqrt{t} \cdot Z+\sqrt{1-t} \cdot D, \]
    as well as the function 
        \[ q(t) := \bbE[Q(G(t))]. \]
    Then for all $0 < t < 1$,
        \[ q'(t) = \frac{1}{2} \sum_{ij} ((\Sigma_Z)_{ij}-(\Sigma_D)_{ij}) \cdot \bbE[\partial_{ij}Q(G(t))]. \]
\end{lemma}

\subsection{Interactive Protocols and Average-Max Circuits}
To model an $\IP$ verifier with oracle access, we use the standard Arthur--Merlin game formulation of~\cite[Section~2]{AGH90}: the verifier's messages consist of fresh uniform random strings of predetermined lengths, and acceptance is determined by a polynomial-time oracle computation on the completed transcript. In a public-coin interactive protocol, postponing the verifier's oracle calls to the end of the protocol  entails no loss of generality. A protocol in which the verifier sends oracle-dependent messages can be compiled into this normal form by asking the prover to reconstruct the verifier's messages from the public randomness, the previous transcript, and the oracle.

It suffices for us to consider polynomial-round Arthur--Merlin protocols, since the private-coin to public-coin transformation of Goldwasser and Sipser~\cite{GS86} relativizes: for every classical oracle $\calO$, $\IP^{\calO} = \AM[\poly]^{\calO}$. Aiello, Goldwasser, and H{\aa}stad~\cite{AGH90} represented Arthur-Merlin protocols by trees of alternating average and maximum gates, corresponding respectively to Arthur's random messages and Merlin's optimal responses. Whereas they only consider protocols in which the verifier makes a single oracle query, we generalize their definition to accommodate a polynomial number of adaptive queries.  

\begin{definition}[AM circuits]
A \emph{$(d,L)$-AM circuit} on $m$ variables is a rooted tree of depth at most $d$, whose internal gates have fan-in at most $2^L$ and compute
\[
    \operatorname{AVG}(z_1,\ldots,z_t)
        = \frac{1}{t}\sum_{j=1}^t z_j,
    \qquad
    \operatorname{MAX}(z_1,\ldots,z_t)
        = \max_{j\in[t]} z_j.
\]
Its leaves are affine functions $\ell(x) = c + \langle u, x\rangle$, where $c \in \bbR$, $u \in \bbR^m$, and $\|u\|_1 \leq L$. Evaluation from the leaves to the root defines a function $C: \bbR^m \to \bbR$.
\end{definition}

The following lemma shows that these circuits indeed capture $\AM[\poly]$ protocols. Fix an $\AM[\poly]$ verifier $V$, where both parties have query access to an oracle $\calO$ represented by $x \in \{\pm 1\}^m$, and define
\[
    a_V(x) := \max_P \Pr[\langle V^x \leftrightarrow P \rangle = 1],
\]
where $\langle V^x \leftrightarrow P \rangle = 1$ means an acceptance event.

\begin{lemma}\label{lemma:am_characterization}
Suppose an $\AM[\poly]$ protocol exchanges $r$ messages, each of a prescribed length at most $L$, and its verifier makes at most $L$ queries to the oracle $x \in \{\pm 1\}^m$.

Then $a_V: \{\pm 1\}^m \to [0,1]$ has an extension $C: \bbR^m \to \bbR$ computable by an $(r+1, L)$-AM circuit.
\end{lemma}
\begin{proof}
As explained above, we may assume that the verifier makes all oracle queries after the interaction. Unroll the interaction into a tree with one level per message: each verifier node is an AVG gate over its possible random messages, and each prover node is a MAX gate over its possible replies. These gates have fan-in at most $2^L$. It remains to represent the verifier's decision at each completed transcript.

Fix such a transcript $\tau$. The verifier's remaining computation is a deterministic oracle decision tree of depth at most $L$: each internal node queries an oracle location, and its two outgoing edges correspond to the possible answers $-1$ and $+1$. A root-to-leaf path $\gamma$ therefore records answers $a_1, \ldots, a_k$ at queried locations $i_1, \ldots, i_k$, where $k \leq L$, and ends with an acceptance bit $c_\gamma \in \{0, 1\}$. Associate with this path the affine function
\[
    \ell_\gamma(x) := c_\gamma - \sum_{j = 1}^k \frac{1-a_jx_{i_j}}{2}.
\]
Its coefficient $\ell_1$-norm is at most $k/2 \leq L/2$. For $x \in \{\pm 1\}^m$, each summand $\frac{1-a_jx_{i_j}}{2}$ is the indicator that the path's answer $a_j$ disagrees with the actual oracle's value $x_{i_j}$. Thus every inconsistent path has score at most zero, whereas the unique consistent path has score equal to the verifier's actual acceptance bit. Consequently, the verifier's decision given $\tau$ equals $\max_\gamma\ell_\gamma(x)$.

At each completed transcript, attach a MAX gate over these affine functions. There are at most $2^L$ paths in the decision tree, so this gate also has fan-in at most $2^L$. The resulting tree has depth at most $r+1$. For each Boolean oracle $x$, averaging over verifier messages and maximizing over prover replies computes $a_V(x)$ by backward induction. Hence this tree defines the desired extension.
\end{proof}

%% file: forrelation.tex
\begin{definition}[Forrelation]
Let $N=2^n$, and let $H = H_N$ be the Hadamard matrix with entries 
    \[ H_{xy} = \frac{1}{\sqrt{N}} (-1)^{x \cdot y}, \qquad x,y\in\{0,1\}^n. \]
For sign vectors $f, g \in \{\pm 1\}^N$, define the \emph{forrelation} between $f$ and $g$ as
\begin{align*}
    \forr(f, g) = \frac{1}{N} f^{\sfT} H_N g = \frac{1}{2^{3n/2}} \sum_{i, j \in \{0, 1\}^n} (-1)^{i \cdot j} f_i g_j.
\end{align*}
\end{definition}
Note that for any $f, g \in \{\pm 1\}^{N}$, $\forr(f, g) \in [-1, 1]$.

Let $\calU_n$ be the uniform distribution on $\{\pm 1\}^{2N}$. We define the distribution $\calD_n$ over $\{\pm 1\}^{2N}$ (first proposed by Aaronson \cite{Aar10}) as follows:
\begin{enumerate}
    \item Sample $x_1, \ldots, x_N \sim \calN(0, 1)$ independently.
    \item Let $y = H_N \cdot x $.
    \item Output $(\sgn(x), \sgn(y))$.
\end{enumerate}
In the rest of our proof, we will need the following fact, which intuitively states that the distributions $\calD_n$ and $\calU_n$ are easy to distinguish for quantum algorithms. Later we will show the same thing does \emph{not} hold for interactive proof systems.
\begin{lemma}\label{lemma:bqp_alg}
There exists an $O(1)$-query $\poly(n)$-time quantum algorithm $\calA$ such that 
    \[ \Pr_{(x, y) \sim \calD_n}\left[\Pr_{\calA}[\calA^{x, y} = 1] \geq \frac{2}{3} \right] \geq 1-O\left(\frac{1}{N}\right), \]
and
    \[ \Pr_{(x, y) \sim \calU_n}\left[\Pr_{\calA}[\calA^{x, y} = 1] \leq \frac{1}{3}\right] \geq 1-O\left(\frac{1}{N}\right). \]
\end{lemma}
The proof of Lemma \ref{lemma:bqp_alg} is deferred to Appendix \ref{sec:bqp_alg} as it is fairly standard.

%% file: separation.tex
\subsection{Structural Result for AM Circuits}
The AM circuit representation allows us to study a verifier's acceptance probability through a convex function, but the presence of MAX gates makes this function non-differentiable, thereby complicating our analysis. To recover differentiability, we replace the MAX gates by smooth approximations, i.e. softmax gates. The following lemma shows that this replacement preserves convexity, and gives explicit bounds on the approximation error as well as gradient and Hessian norms in terms of the circuit parameters.

\begin{lemma}[Softmax smoothing]\label{lemma:structural}
Let $C$ be a $(d, L)$-AM circuit which takes inputs $x \in \{\pm 1\}^m$. Then for every $\lambda > 0$, there is a smooth convex function $F: \bbR^m \to \bbR$ such that $0 \leq F-C \leq \frac{dL \ln 2}{\lambda}$, $\sup_x \|\nabla F(x)\|_1 \leq L$, and $\sup_x \|\nabla^2 F(x)\|_1 \leq 2\lambda dL^2$.
\end{lemma}
\begin{proof}
Begin by replacing each MAX gate in $C$ by the softmax gate
    \[ S_\lambda(u_1, \ldots, u_k) = \lambda^{-1}\ln(\sum_{i = 1}^k e^{\lambda u_i}). \]
Observe that
    \[ \max(u_1, \ldots, u_k) \leq S_\lambda(u_1, \ldots, u_k) \leq \max(u_1, \ldots, u_k)+\frac{\ln k}{\lambda} \leq \max(u_1, \ldots, u_k)+\frac{L \ln 2}{\lambda}. \]

For a gate $v$ in the circuit, denote by $C_v(x)$ the value of the gate on input $x$ in the original circuit $C$ and by $F_v(x)$ the value of the corresponding gate in the new circuit with softmax gates.

We show by induction on the MAX-gate height $h$ that for every gate $v$ at MAX-gate height $h$, $0 \leq F_v(x)-C_v(x) \leq hL\ln(2)/\lambda$. The claim is obvious at height $h = 0$, and also clearly follows for any AVG gate, so it remains to prove it for MAX gates at height $h \geq 1$.

For $h \geq 1$, let $v$ be a MAX gate at height $h$ and let $u_1, \ldots, u_k$ be the gates, each of height at most $h-1$, entering the gate $v$. On the one hand, 
    \[ C_v(x) = \max(C_{u_1}(x), \ldots, C_{u_k}(x))  \leq \max(F_{u_1}(x), \ldots, F_{u_k}(x)) \leq S_{\lambda}(F_{u_1}(x), \ldots, F_{u_k}(x)) = F_v(x). \]
On the other hand,
\begin{align*}
    F_v(x) &= S_{\lambda}(F_{u_1}(x), \ldots, F_{u_k}(x)) \\
    &\leq \max(F_{u_1}(x), \ldots, F_{u_k}(x)) + \frac{L\ln2}{\lambda} \\
    &\leq \max(C_{u_1}(x) + \tfrac{(h-1)L\ln2}{\lambda}, \ldots, C_{u_k}(x) + \tfrac{(h-1)L\ln2}{\lambda}) + \frac{L\ln2}{\lambda} \\
    &= \max(C_{u_1}(x), \dots, C_{u_k}(x)) + \frac{hL\ln2}{\lambda} \\
    &= C_v(x) + \frac{hL\ln2}{\lambda},
\end{align*}
which concludes the inductive proof. Smoothness follows from the fact that all gates (softmax, AVG, affine leaves) are smooth, and smoothness is preserved by composition.

For child functions $u_1, \ldots, u_k$, writing $p_i = \frac{e^{\lambda u_i}}{\sum_j e^{\lambda u_j}}$ we have that
\begin{align*}
    \nabla S_\lambda &= \sum_i p_i \nabla u_i.
\end{align*}
Hence, if every child $u_i$ satisfies $\|\nabla u_i\|_1 \leq L$, then
    \[ \|\nabla S_\lambda\|_1 \leq \sum_i p_i \|\nabla u_i\|_1 \leq L. \]
Likewise, at an AVG gate,
    \[ \operatorname{AVG}(u_1, \ldots, u_k) = \frac{1}{k} \sum_i u_i, \]
and thus
    \[ \nabla \operatorname{AVG} = \frac{1}{k} \sum_i \nabla u_i, \]
which means the gradient norm bound is also preserved by AVG gates. Starting from the affine leaves, which have $\ell_1$-gradient norm at most $L$, proves that $\sup_x \|\nabla F(x)\|_1 \leq L$.

Differentiating $\nabla S_{\lambda}$ gives
    \[ \nabla^2 S_\lambda = \sum_i p_i \nabla^2 u_i + \sum_i \nabla u_i (\nabla p_i)^\sfT. \]
To evaluate this, we first determine the gradient of $p_i$: 
\begin{align*}
    \nabla p_i &= \frac{(\sum_j e^{\lambda u_j}) \cdot \lambda  e^{\lambda u_i} \cdot \nabla u_i - e^{\lambda u_i} \cdot \sum_j (\lambda e^{\lambda u_j} \cdot \nabla u_j)}{(\sum_r e^{\lambda u_r})^2} \\
    &= \frac{\lambda e^{\lambda u_i} \cdot \nabla u_i}{\sum_r e^{\lambda u_r}} - \frac{\lambda e^{\lambda u_i}}{\sum_r e^{\lambda u_r}} \cdot \frac{\sum_j (e^{\lambda u_j} \cdot \nabla u_j)}{\sum_r e^{\lambda u_r}} \\
    &= \lambda p_i \nabla u_i - \lambda p_i \cdot \left(\sum_j p_j \nabla u_j\right).
\end{align*}
Thus,
\begin{align*}
    \nabla^2 S_\lambda &= \sum_i p_i \nabla^2 u_i + \sum_i  \nabla u_i (\nabla p_i)^\sfT \\
    &= \sum_i p_i \nabla^2 u_i + \sum_i \lambda p_i \nabla u_i (\nabla u_i - \sum_j p_j \nabla u_j)^\sfT \\
    &= \sum_i p_i \nabla^2 u_i + \lambda \left[\sum_i p_i \nabla u_i (\nabla u_i)^\sfT - \sum_{i, j} p_i p_j \nabla u_i (\nabla u_j)^\sfT\right] \\
    &= \sum_i p_i \nabla^2 u_i + \frac{\lambda}{2} \sum_{i, j} p_i p_j (\nabla u_i - \nabla u_j)(\nabla u_i - \nabla u_j)^\sfT.
\end{align*}
The last term is positive semidefinite, which proves convexity by induction on the depth of $C$, as the leaves of $C$ are convex.

To bound the Hessian norm, we note that
\begin{align*}
    \|\nabla^2 S_\lambda\|_1 &\leq \sum_i p_i \|\nabla^2 u_i\|_1 + \frac{\lambda}{2} \sum_{i, j} p_i p_j \|(\nabla u_i - \nabla u_j)(\nabla u_i - \nabla u_j)^\sfT\|_1 \\
    &\leq \sum_i p_i \|\nabla^2 u_i\|_1 + \frac{\lambda}{2} \sum_{i, j} p_i p_j \|\nabla u_i - \nabla u_j\|_1^2 \\
    &\leq \sum_i p_i \|\nabla^2 u_i\|_1 + \frac{\lambda}{2} \sum_{i, j} p_i p_j (2L)^2 = \sum_i p_i \|\nabla^2 u_i\|_1 + 2\lambda L^2,
\end{align*}
so softmax gates increase the Hessian norm by at most $2\lambda L^2$. Meanwhile, for an AVG gate,
    \[ \|\nabla^2 \operatorname{AVG}\|_1 \leq \frac{1}{k} \sum_i \|\nabla^2 u_i\|_1, \]
so average gates do not increase the Hessian norm. Since the leaves have zero Hessian norm, we have by induction on the depth of $C$ that $\sup_x \|\nabla^2 F(x)\|_1 \leq 2\lambda dL^2$.
\end{proof}

\begin{remark}

Convexity is essential to our use of the derivative bounds. Indeed, Appendix \ref{sec:smooth-extensions} shows that every function \(f:\{\pm1\}^m\to[-1,1]\) admits a smooth extension \(F:\mathbb R^m\to[-1,1]\) agreeing exactly with \(f\) on the Boolean cube and satisfying
\[
\sup_x\|\nabla F(x)\|_1\le2,
\qquad
\sup_x\|\nabla^2F(x)\|_1\le C\log(2m),
\]
where the Hessian norm is the entrywise \(\ell_1\)-norm and \(C\) is an absolute constant. Thus even a dimension-independent gradient bound and a logarithmic Hessian bound are compatible with arbitrary Boolean behavior. The crucial property of our construction is achieving the required derivative control simultaneously with convexity.
\end{remark}

\subsection{Notation}\label{section:notation}
In the following section, we will use $\delta$-smoothings of the function $f_{a}(x) := \sgn(a+x)$ for any fixed $a \in \bbR^m$. Specifically, for any $\delta \in (0, 1]$ and real vector $a \in \bbR^m$, we define the functions
\begin{align*}
    h_{a, \delta, i}(t) &:= 2\Phi\left(\frac{a_i + \sqrt{1-\delta^2} t}{\delta}\right)-1, &\text{for $i\in [m]$}\\
    h_{a, \delta}(x) &:= (h_{a, \delta, 1}(x_1), \ldots, h_{a, \delta, m}(x_m)).
\end{align*}
Observe that $h_{a, \delta, i}(t) \in [-1, 1]$ and 
\begin{align*}
    |h_{a, \delta, i}'(t)| = \left|\tfrac{2\sqrt{1-\delta^2}}{\delta}\cdot 
       \varphi\left(\tfrac{a_i + \sqrt{1-\delta^2}t}{\delta}\right)\right| \leq \sqrt{\frac{2}{\pi}} \cdot \frac{1}{\delta} \leq \frac{1}{\delta}.
\end{align*}
Further, observe that there is another way to derive $h_{a,\delta}(x)$ as \[h_{a,\delta}(x) = \E[\sgn(a + \sqrt{1-\delta^2}x + \delta G')]\] for any MVG $G'$ of dimension $m$ with zero means and variances $1$ (and arbitrary covariances as we only care about the expectation coordinate-wise).

\subsection{Forrelated Gaussians Fool AM Circuits}

The following lemma provides the key technical insight required in the proof of our main result. The lemma bounds the change in the expected value of a
smooth convex function $F$ when Gaussian sign rounding is replaced
by a smooth approximation: it shows that
$\bbE[F(\sgn(a+G))]$ and
$\bbE[F(h_{a,\delta}(G))]$ are close under the stated assumptions. Here, $G$ is a nondegenerate Gaussian vector with standard normal marginals, and $h_{a,\delta}$ is the coordinate-wise smooth approximation to sign rounding defined in Section~\ref{section:notation}.

The error is controlled by the $\ell_1$-gradient bound of $F$ and the smoothing parameter $\delta$, with only a logarithmic dependence on the input dimension $m$. Crucially, this lemma uses the convexity of $F$ to get a better bound than would be implied by a naive application of H\"older's inequality. In our application, $F$ is the smooth approximation to an AM circuit supplied by Lemma~\ref{lemma:structural}. 

\begin{lemma}\label{lemma:smoothening}
Let $F: \bbR^m \to \bbR$ be a smooth convex function where
    \[ \sup_{x \in [-1, 1]^m} \norm{\nabla F(x)}_1 \leq L. \]
Let $D$ be an $m$-dimensional MVG vector with zero means, variances $1$ and arbitrary covariances.
Let $Z\sim \calN(0,I_{m})$ be independent of $D$, and let $G = \sqrt{\tau} Z + \sqrt{1-\tau}D$ for some $\tau \in [1/2,1]$.
Let $a\in \bbR^m$ be any vector and $\delta \in (0, 1/2]$.
Then,
\[ 
0 \leq \E[F(\sgn(a+G))]-\E[F(h_{a,\delta}(G))] \leq 
18L\delta \cdot (\|a\|_{\infty}+\sqrt{\ln 2m}).\]
\end{lemma}
\begin{proof}
We first prove that the lower bound follows from the convexity of $F$.
Indeed, let $G'$ be an MVG with the same law as $G$, drawn independently of $G$. 
Then, by the definition of $h_{a, \delta}$ and Jensen's inequality
we have 
\[
\E[F(h_{a,\delta}(G))]
= \E[F(\E[\sgn(a+\sqrt{1-\delta^2} G + \delta G') \mid G])]
\le \E[F(\sgn(a+\sqrt{1-\delta^2} G + \delta G'))],
\]
but the distribution of $a+\sqrt{1-\delta^2} G + \delta G'$ is exactly the same as that of $a+G$, thus the RHS equals $\E[F(\sgn(a+G))]$, completing the lower bound.

It remains to prove the upper bound.  
Define the random variable $S := \sgn(a+G)$. 
Since $F$ is smooth and convex, we have that everywhere
\begin{equation}F(S)-F(h_{a,\delta}(G)) \leq \nabla F(S) \cdot \bigl(S-h_{a,\delta}(G)\bigr).\end{equation}

The key point is that the RHS is the inner product of (i) a vector $\nabla F(S)$ with $\ell_1$-norm at most $L$ \emph{that only depends on $S$}, with (ii) another vector $(S-h_{a,\delta}(G))$ whose entries are typically small (of the order of $\delta$) where each entry still has a lot of randomness even conditioned on $S$.

So imagine we condition on $S$ and $D$. For $i \in [m]$, the $i$th coordinate of $G$ conditioned on $(S, D)$ is distributed as a truncated 1-dimensional Gaussian. We claim that after this conditioning, $h_{a,\delta, i}(G_i)$ is (on average) very close to $S_i$.
\begin{claim}\label{claim:bounded_deviation}
    $|S_i - \E[h_{a,\delta, i}(G_i) | S_i,D_i]| \leq 3\delta(1+|a_i|+|D_i|)$.
\end{claim}
Before we prove this claim, let's see why it finishes the proof. Indeed, assuming Claim \ref{claim:bounded_deviation}, then we have that
\begin{align*}
\E[F(S)-F(h_{a, \delta}(G))] &\leq \E[\nabla F(S) \cdot \bigl(S-h_{a,\delta}(G)\bigr)] \\
&= \E_{S,D}[\E[\nabla F(S) \cdot \bigl(S-h_{a,\delta}(G)\bigr)|S,D]]\\
&\leq \E_{S,D}[\nabla F(S) \cdot \E[\bigl(S-h_{a,\delta}(G)\bigr)|S,D]]\\
&\leq \E_{S,D}[\|\nabla F(S)\|_1 \cdot \|\E[\bigl(S-h_{a,\delta}(G)\bigr)|S,D]\|_{\infty}]\\
&\leq \E_{S,D}[\|\nabla F(S)\|_1 \cdot \max_i\{3\delta(1+|a_i|+|D_i|)\}]\\
&\leq L\cdot 3\delta \cdot (1+\|a\|_{\infty}+\E[\max_{i}|D_i|])\\
&\leq 3L\delta \cdot \bigl(1+\|a\|_{\infty}+4\sqrt{\ln 2m}\bigr) \tag{Fact~\ref{fact:exp_max}}\\
&\leq 18L\delta \cdot \bigl(\|a\|_{\infty}+\sqrt{\ln 2m}\bigr).\qedhere
\end{align*}
\end{proof}

\begin{proof}[Proof of Claim~\ref{claim:bounded_deviation}]
Let $Z_i\sim\mathcal{N}(0,1)$, let $\tau\in[1/2,1]$,
and let $d_i,a_i\in\mathbb{R}$ and $s_i\in \{-1,1\}$ be fixed. Define the random variable $Z'_i := \sqrt{\tau}\,Z_i+\sqrt{1-\tau}\,d_i+a_i$. Note that $Z'_i\sim\mathcal{N}(\mu,\sigma^2)$, where $\mu=\sqrt{1-\tau}\,d_i+a_i$ and $\sigma = \sqrt{\tau}$.

We need to show that
    \[ \bigl|s_i-\E[h_{a,\delta,i}(Z'_i-a_i)\,|\,\sgn(Z'_i)=s_i]\bigr| \leq 3\delta(1+|d_i|+|a_i|).\]
Without loss of generality, assume $s_i = 1$ (as negating $a,d,s$ maintains the same absolute value). Our new goal is to show that
\begin{equation}\label{eq:new_goal}
\E[1-h_{a,\delta,i}(Z'_i-a_i) | Z'_i>0]\le 3\delta (1+|d_i|+|a_i|).
\end{equation}
Let $p_+$ denote the density of $Z'_i$ conditioned on $Z'_i>0$.

To prove \Cref{eq:new_goal}, we first move to the following symmetric function
\[
    q_\delta(t)
    :=2\Phi\!\left(\frac{\sqrt{1-\delta^2}\,t}{\delta}\right)-1.
\]
Since $2\Phi-1$ is $\sqrt{2/\pi}$-Lipschitz, for every $t\in\bbR$,
\begin{align*}
    \left|h_{a,\delta,i}(t-a_i)-q_\delta(t)\right|
    &\le \sqrt{\frac{2}{\pi}}\,
         \frac{(1-\sqrt{1-\delta^2})|a_i|}{\delta}= \sqrt{\frac{2}{\pi}}\,
       \frac{\delta|a_i|}{1+\sqrt{1-\delta^2}}
    \le \delta|a_i|.
\end{align*}
Moreover, $1-q_\delta(t)\ge0$ for $t>0$, and
\[
    \int_0^\infty (1-q_\delta(t))\,dt
    =\frac{2\delta}{\sqrt{1-\delta^2}}
      \int_0^\infty(1-\Phi(u))\,du
    =\sqrt{\frac{2}{\pi}}\frac{\delta}{\sqrt{1-\delta^2}}
    \le\delta.
\]
Consequently, 
\begin{align*}
    \E\!\left[
        1-h_{a,\delta,i}(Z'_i-a_i)
        \,\middle|\,Z'_i>0
    \right]
    &\le \E[1-q_\delta(Z'_i)\mid Z'_i>0]
         +\delta|a_i|\\
    &\le \delta\sup_{t>0}p_+(t)+\delta|a_i|.
\end{align*}
We next prove that
\[ \sup_{t>0}p_+(t) \le \frac{1}{\sigma}+\frac{|\mu|}{\sigma^2}, \]
which will complete the proof as $\frac{1}{\sigma}+\frac{|\mu|}{\sigma^2} \le 2(1+|d_i|+|a_i|)$.

We split into two cases. If $\mu \ge 0$, then $\Pr[Z'_i>0]\ge 1/2$, and hence
\[
  p_+(t)
    \le 
    \frac{\exp\!\left(-(t-\mu)^2/(2\sigma^2)\right)}
         {\sigma\sqrt{2\pi}\,\Pr[Z'_i>0]}
    \le \frac{\sqrt{2/\pi}}{\sigma}
    \le \frac{1}{\sigma} \le \frac{1}{\sigma} + \frac{|\mu|}{\sigma^2}.
\]

Otherwise, $\mu<0$, and the conditional density $p_+$ is decreasing on $(0,\infty)$, with
\[
    \sup_{t>0} p_+(t)
    = \frac{1}{\sigma}
      \frac{e^{-b^2/2}}{\int_b^\infty e^{-t^2/2}\,dt}
\]
for $b=|\mu|/\sigma$. 
We claim that $\frac{e^{-b^2/2}}{\int_b^\infty e^{-t^2/2}\,dt}\le  b+1$.
Indeed, 
\[
\frac{e^{-b^2/2}}{b+1}
= \int_b^\infty -\frac{d}{dt}\!\left(\frac{e^{-t^2/2}}{t+1}\right)\,dt
= \int_b^\infty e^{-t^2/2}
  \underbrace{\frac{t^2+t+1}{(t+1)^2}}_{\le 1}\,dt
\le \int_b^\infty e^{-t^2/2}\,dt.
\]
Overall, in this case we also derive
\[
\sup_{t>0}p_+(t) \le \frac{b+1}{\sigma} = \frac{1}{\sigma} + \frac{|\mu|}{\sigma^2},
\]
which completes the proof.
\end{proof}

We are now ready to prove our main theorem, which shows that multivariate Gaussians with sufficiently weak covariance are indistinguishable from each other to smooth convex functions with low norms. 
\begin{theorem}\label{theorem:comparison}
Let $F: \bbR^{m} \to \bbR$ be a smooth convex function where
    \[ \sup_{x \in [-1, 1]^{m}} \norm{\nabla F(x)}_1 \leq L, \qquad \sup_{x \in [-1, 1]^{m}} \norm{\nabla^2F(x)}_1 \leq K, \]
where $L \leq m$. Let $D \sim \calN(0, \Sigma)$ be an $m$-dimensional Gaussian vector where $\Sigma_{ii} = 1$ and $\beta := \max_{i\neq j}|\Sigma_{ij}|$. Let $Z \sim \calN(0, I_{m})$ be an independent MVG, and let $G = \sqrt{1/2} D + \sqrt{1/2} Z$. Then, for every $a \in \bbR^{m}$ and $0 < \delta \leq 1/2$,
\begin{align*}
    \bigg|\E[F(\sgn(a+G))]-\E[F(\sgn(a+Z))]\bigg| \leq 400 L\delta \sqrt{\ln(2m)} + \frac{\beta K}{4\delta^2} + \frac{4}{m^{48}}\;.
\end{align*}
\end{theorem}
\begin{proof}[Proof of Theorem \ref{theorem:comparison}]
Fix $a$ and $\delta$. First, observe that by H\"older's inequality, we know that for all $x, y \in [-1, 1]^{m}$,
\begin{align*}
    |F(x)-F(y)| &= \left|\int_0^1 \nabla F(x+t(y-x)) \cdot (y-x) \,dt\right| \\
    &\leq \sup_{t \in [0, 1]} \|\nabla F(x+t(y-x))\|_1 \cdot \|y-x\|_{\infty} \leq 2L.
\end{align*}
For each coordinate $i$ where $|a_i| > 10\sqrt{\ln m}$, replace the $i$th coordinate of $\sgn(a+G)$ and of $\sgn(a+Z)$ with $\sgn(a_i)$. Each coordinate of $G$ and of $Z$ is standard normal, so
    \[ \Pr[\sgn(a_i+G_i) \neq \sgn(a_i)] \leq \Phi(-|a_i|) \leq e^{-a_i^2/2} \leq m^{-50}, \]
and the same bound holds for $Z$. A union bound and the Lipschitz bound $|F(x)-F(y)|\le 2L$ imply that each of $\E[F(\sgn(a+G))]$ and $\E[F(\sgn(a+Z))]$ changes by at most $2Lm^{-49} \le 2m^{-48}$.

After having fixed all coordinates where $|a_i| > 10\sqrt{\ln m}$, it is straightforward to see that $F$ remains a smooth convex function with the same norm bounds. Similarly, the covariance matrix of (the remaining coordinates of) $G$ satisfies the same conditions as before. It thus suffices to show that 
\begin{align*}
    \bigg|\E[F(\sgn(a+G))]-\E[F(\sgn(a+Z))]\bigg| \leq 400 L\delta \sqrt{\ln(2m)} + \frac{\beta K}{4\delta^2},
\end{align*}
assuming $|a_i| \leq 10\sqrt{\ln m}$ for all $i \in [m]$.

To interpolate between $Z$ and $G$ we introduce the following.
For $1/2 \leq \tau \leq 1$, denote the random variable $G(\tau) = \sqrt{\tau}Z+\sqrt{1-\tau}D$. Thus $G(1/2) \sim G$ while $G(1) \sim Z$. Our ultimate goal is (more or less) to smoothly interpolate between these endpoints under the function $F$.

The proof consists of two major steps. We begin by approximating $\sgn(a+G(\tau))$ by $h_{a,\delta}(G(\tau))$. By Lemma \ref{lemma:smoothening}, this smoothing is a good approximation of both endpoints under $F$, provided $F$ has low $\ell_1$-gradient norm.

Our second step is to show that once we have smoothened $\sgn(a+G(\tau))$ to $h_{a,\delta}(G(\tau))$, the distributions $G(1/2)$ and $G(1)$ are similar from the perspective of $F$:

\begin{lemma}\label{lemma:interpolation}
$\left|\E[F(h_{a,\delta}(G(1)))]-\E[F(h_{a,\delta}(G(1/2)))]\right| \leq \frac{\beta K}{4\delta^2}$.
\end{lemma}
\begin{proof}[Proof of Lemma \ref{lemma:interpolation}]
Define the function  $Q := F \circ h_{a,\delta}$ (which satisfies the conditions of Lemma \ref{lemma:smart_path} since all partial derivatives of $h_{a,\delta}$ are bounded for fixed $a, \delta$ and $F$ is smooth on $[-1, 1]^{m}$). 
Lemma \ref{lemma:smart_path} then gives
for $q(t) = \bbE[Q(G(t))]$ that 
\begin{align*}
    q'(t)
    &= \frac{1}{2} \sum_{i, j \in [m]}(I - \Sigma)_{ij}
          \E[\partial_{ij}Q(G(t))]\\
    &= \frac{1}{2} \sum_{i \neq j \in [m]} (-\Sigma_{ij})
          \E[\partial_{ij}Q(G(t))].
\end{align*}
For $i \neq j$, the chain rule gives
    \[ \partial_{ij} Q(x) = (\partial_{ij}F)(h_{a,\delta}(x)) \cdot h_{a,\delta,i}'(x_i) h_{a,\delta,j}'(x_j). \]
Thus,
\begin{align*}
    |q'(t)| &\leq \frac{\max_{i \neq j} |\Sigma_{ij}|}{2} \bbE\left[\sum_{i \neq j}|\partial_{ij}Q(G(t))|\right] = \frac{\beta}{2} \bbE\left[\sum_{i \neq j}|\partial_{ij}Q(G(t))|\right] \\
    &\leq\frac{\beta}{2\delta^2} \bbE\left[\sum_{i \neq j}|(\partial_{ij}F)(h_{a,\delta}(G(t)))|\right] \leq \frac{\beta K}{2\delta^2}.
\end{align*}
Integrating over $1/2 \leq t \leq 1$ and taking absolute values gives
    \[ \left|\E[F(h_{a,\delta}(G(1)))]-\E[F(h_{a,\delta}(G(1/2)))] \right| \leq \int_{1/2}^1 |q'(t)|\,dt \leq \frac{\beta K}{4\delta^2}, \]
as claimed.
\end{proof}

As Lemma \ref{lemma:smoothening} implies that for $\tau \in [1/2, 1]$,
    \[ 0 \leq \E[F(\sgn(a+G(\tau)))] - \E[F(h_{a,\delta}(G(\tau)))] \leq 
    18L\delta \cdot (\|a\|_{\infty}+\sqrt{\ln(2m)}) \le 200 L\delta \sqrt{\ln(2m)}, \]
we conclude that
\begin{align*}
    \bigg|\E[F(\sgn(a+G))]-\E[F(\sgn(a+Z))]\bigg| &= \bigg|\E[F(\sgn(a+G(1/2)))]-\E[F(\sgn(a+G(1)))]\bigg| \\
    &\leq \bigg|\E[F(\sgn(a+G(1/2)))]-\E[F(h_{a,\delta}(G(1/2)))]\bigg| 
    \\&\quad+ \bigg|\E[F(h_{a,\delta}(G(1/2)))]-\E[F(h_{a,\delta}(G(1)))]\bigg| 
    \\&\quad+ \bigg|\E[F(h_{a,\delta}(G(1)))]- \E[F(\sgn(a+G(1)))]\bigg|\\
    &\leq 400 L\delta \sqrt{\ln(2m)} + \frac{\beta K}{4\delta^2}
\end{align*}
\end{proof}

Next, we show that in fact one can compare $F(\sgn(D))$ to $F(\sgn(Z))$ directly without needing to apply noise to $D$. This is achieved by iterative application of \Cref{theorem:comparison}.

\begin{corollary}\label{cor:recursive}
Let $F: \bbR^{m} \to \bbR$ be a smooth convex function where
    \[ \sup_{x \in [-1, 1]^{m}} \norm{\nabla F(x)}_1 \leq L, \qquad \sup_{x \in [-1, 1]^{m}} \norm{\nabla^2F(x)}_1 \leq K, \]
where $L \leq m$. Let $D \sim \calN(0,\Sigma)$ be an $m$-dimensional Gaussian vector where $\Sigma_{ii} = 1$ and $\beta := \max_{i\neq j}|\Sigma_{ij}|$. Let $Z \sim \calN(0,I_{m})$ be an independent MVG. Then, for every $0 < \delta \leq 1/2$ and $T \in \bbN$,
\begin{align*}
    \bigg|\E[F(\sgn(D))]-\E[F(\sgn(Z))]\bigg| \leq T \cdot \left(400 L\delta \sqrt{\ln(2m)} + \frac{\beta K}{4\delta^2} + \frac{4}{m^{48}}\right) + 2^{-T/2} \cdot mL\;.
\end{align*}
\end{corollary}
\begin{proof}
Recall the definition $G(\tau) = \sqrt{\tau} Z + \sqrt{1-\tau}D$. We apply \Cref{theorem:comparison} iteratively to compare $G(0)$ and $G(1)$. For $i=0, 1,\ldots,T$, define the hybrid distribution $H^{(i)} := G(2^{-i})$ and define $H^{(T+1)} = G(0)$. We first note that
    \[ \E[F(\sgn(G(1)))] - \E[F(\sgn(G(0)))] =  \sum_{i=0}^{T} \left(\E[F(\sgn(H^{(i)}))] - \E[F(\sgn(H^{(i+1)}))]\right) \]
by a simple telescoping identity.

We begin by showing that for $0 \leq i \leq T-1$, 
    \[ \bigl|\E[F(\sgn(H^{(i)}))] - \E[F(\sgn(H^{(i+1)}))]\bigr|  \le 400 L\delta \sqrt{\ln(2m)} + \frac{\beta K}{4\delta^2} + \frac{4}{m^{48}}. \]
The case $i = 0$ is exactly comparing $G(1)$ to $G(1/2)$, which is proved in \Cref{theorem:comparison}.

For $i > 0$, note that
    \[ H^{(i)} \sim \sqrt{2^{-i}}Z + \sqrt{1-2^{-i}}D, \] 
and
    \[ H^{(i+1)} \sim \sqrt{2^{-(i+1)}}Z + \sqrt{2^{-(i+1)}}D' + \sqrt{1-2^{-i}}D, \]
where $D'$ is an independent copy of $D$. Here we use the fact that the sum of independent MVGs is an MVG. Now observe that we can think of first fixing $D = a$ and then comparing $\E[F(\sgn(H^{(i+1)}))|D=a]$ with $\E[F(\sgn(H^{(i)}))|D=a]$. But conditioned on $D = a$,
    \[ \sgn(H^{(i)}) = \sgn(\sqrt{2^{i}} H^{(i)}) = \sgn\left(Z + \sqrt{2^{i}(1-2^{-i})}a\right), \]
and
    \[ \sgn(H^{(i+1)}) = \sgn(\sqrt{2^{i}} H^{(i+1)}) = \sgn\left(\sqrt{1/2}Z +\sqrt{1/2}D' + \sqrt{2^{i}(1-2^{-i})}a\right), \]
so we see that we have again reduced to the case of \Cref{theorem:comparison} but now with the shift vector $\sqrt{2^{i}(1-2^{-i})}a$. By averaging over $a$ we get
\begin{align*}
\bigl|\E[F(\sgn(H^{(i+1)}))] - \E[F(\sgn(H^{(i)}))]\bigr| \le 400 L\delta \sqrt{\ln(2m)} + \frac{\beta K}{4\delta^2} + \frac{4}{m^{48}}.
\end{align*}

Finally, we compare $H^{(T)} = G(2^{-T})$ with $H^{(T+1)} = G(0)$. Indeed, for each $i \in [m]$,
    \[ H^{(T)}_i = \sqrt{2^{-T}} Z_i + \sqrt{1-2^{-T}}D_i,\qquad H^{(T+1)}_i=D_i\;. \]
Their correlation is $\sqrt{1-2^{-T}}$ and thus
    \[ \E[\sgn(H^{(T)}_i) \sgn(H^{(T+1)}_i)] = \tfrac{2}{\pi} \arcsin(\sqrt{1-2^{-T}}) \ge 1-2^{-T/2}, \]
and so the signs disagree with probability at most $2^{-T/2-1}$. Taking a union bound over all $m$ coordinates, we see that with probability at most $m2^{-T/2-1}$, $\sgn(H^{(T)}) \neq \sgn(H^{(T+1)})$. Furthermore, whenever the two vectors differ, $|F(\sgn(H^{(T)}))-F(\sgn(H^{(T+1)}))|\le 2L$ by the $\ell_1$-gradient norm of $F$.
Thus,
    \[ \bigl|\E[F(\sgn(H^{(T)}))]-\E[F(\sgn(H^{(T+1)}))]\bigr| \leq mL\,2^{-T/2}\;.\]
Summing the $T$ hybrid steps and this final comparison, and using $H^{(0)}=Z$ together with $H^{(T+1)}=D$, gives
\[
    \bigl|\E[F(\sgn(Z))]-\E[F(\sgn(D))]\bigr|
    \leq T\cdot\left(400 L\delta\sqrt{\ln(2m)}+\frac{\beta K}{4\delta^2}+\frac{4}{m^{48}}\right)+mL\cdot 2^{-T/2}.
\]
\end{proof}

Finally, we achieve the main result of this section: the Forrelation distribution $\calD_n$ is pseudorandom against $(d,L)$-AM circuits with  $d, L = \poly(n)$.

\begin{corollary}\label{cor:values}
For any $(d, L)$-AM circuit $C$ where $d, L = \poly(n)$,
    \[ \left|\E_{(x, y) \sim \calD_n}[C(x, y)] - \E_{(x, y) \sim \calU_n}[C(x, y)]\right| \leq 2^{-\Omega(n)}. \]
\end{corollary}
\begin{proof}
Set $\delta = N^{-1/8}$,  $\lambda = 1/\delta$, and $T = \lceil 10 \ln(2N) \rceil$.
Applying \Cref{lemma:structural} gives a smooth convex function $F$ such that $0 \leq F-C \leq \frac{dL \ln 2}{\lambda}$, $\sup_x \|\nabla F(x)\|_1 \leq L$, and $\sup_x \|\nabla^2 F(x)\|_1 \leq 2\lambda dL^2$. 
 
Note that $\sgn(Z) \sim \calU_n$ where $Z \sim \calN(0,I_{2N})$ and $\calD_n$ is distributed as the signs of a multivariate Gaussian with covariance matrix 
    \[ \Sigma = \begin{bmatrix} I_N & H_N \\ H_N & I_N \end{bmatrix}. \]
As the off-diagonal entries of $\Sigma$ have magnitude at most $\beta = 1/\sqrt{N}$, applying \Cref{cor:recursive} with $m=2N$ gives
\begin{align*}
    &\bigl|\E_{(x, y) \sim \calD_n}[C(x, y)] - \E_{(x, y) \sim \calU_n}[C(x, y)]\bigr| \\
    \leq \,\, &\bigl|\E_{(x, y) \sim \calD_n}[F(x, y)] - \E_{(x, y) \sim \calU_n}[F(x, y)]\bigr| + \frac{dL \ln 2}{\lambda} \\
    \leq \,\, &\frac{dL \ln 2}{\lambda} + T\cdot \left(400 L\delta \sqrt{\ln 2m} + \frac{\beta \cdot (2\lambda dL^2)}{4\delta^2} +  \frac{4}{m^{48}}\right)  + mL\cdot 2^{-T/2}\\
    = \,\, &\frac{\poly(n)}{2^{n/8}} = 2^{-\Omega(n)}.
\end{align*} 
\end{proof}

\begin{corollary}[Reminder of Theorem \ref{theorem:oracle_separation}]\label{cor:separation}
There exists a classical oracle $\calO$ relative to which $\BQP^{\calO} \not\subseteq \IP^{\calO}$.
\end{corollary}
The proof of Corollary \ref{cor:separation} follows by a straightforward diagonalization argument, which we include in Appendix \ref{sec:diag} for completeness.

%% file: prg.tex
In this section, we generalize \Cref{cor:values} to any centered MVG distribution with variance $1$ and small covariances.
\begin{corollary}\label{cor:small_beta}
Let $m, d, L \in \bbN$, $\beta > 0$ with $L \le m$. Let $G \sim \calN(0, \Sigma)$, with $\Sigma_{i, i} = 1$ and $|\Sigma_{i, j}| \leq \beta$ for $i \neq j$. Then, for any $(d, L)$-AM circuit $C$,
    \[ \bigl|\E[C(\sgn(G))] - \E[C(\calU_{m})]\bigr| \leq \beta^{1/4} \cdot O(dL^2 \cdot \log^{3/2}(mL)) + \tfrac{1}{(2mL)^4}. \]
\end{corollary}
\begin{proof}
Set $\delta = \beta^{1/4}$, $\lambda = 1/\delta$, and $T = \lceil 10 \log_2(2mL) \rceil$. The same calculation as in \Cref{cor:values} gives \begin{align*}
    \bigl|\E[C(\sgn(G))] - \E[C(\calU_{m})]\bigr| 
    \leq \,\, &\frac{dL \ln 2}{\lambda} + T\cdot \left(400 L\delta \sqrt{\ln(2m)} + \frac{\beta \cdot (2\lambda dL^2)}{4\delta^2} +  \frac{4}{m^{48}}\right)  + mL\cdot 2^{-T/2}\\
    = \,\, &\beta^{1/4} \cdot O(dL^2 \cdot \log^{3/2}(mL)) + \tfrac{1}{(2mL)^4}\;.\qedhere
\end{align*} 
\end{proof}

The work of \cite{CHLT19}, which relied on the explicit small-biased spaces of Ta-Shma \cite{TaShma17}, showed how to $1/\poly(m)$-approximately sample an $m$-dimensional centered Gaussian vector $G$ with variances $1$ and covariances at most $\beta$ using seed length $O(\log^2 m)/\beta^{2+o(1)}$. We spell it out next.

\paragraph{From Balanced Codes to Low-Rank MVG with Small Covariances.}
A binary linear code $\mathcal B\subseteq\{0,1\}^r$ is \emph{$\beta$-balanced} if every nonzero codeword has relative Hamming weight in $[(1-\beta)/2,(1+\beta)/2]$. For $k\ge \log_2(m)$, choose a $k$-dimensional $\beta$-balanced code. Take distinct codewords $c_1,\ldots,c_m$ and define their corresponding unit vectors
\[
    v_i=\frac{1}{\sqrt r}
       \bigl((-1)^{(c_i)_1},\ldots,(-1)^{(c_i)_r}\bigr)
\]
for $i\in [m]$. Then $\|v_i\|_2=1$, and for $i\neq j$, $\beta$-balance gives $|\langle v_i,v_j\rangle|\leq\beta$.

Ta-Shma's explicit construction~\cite{TaShma17} gives such a code with
\[
    r=O\!\left(\frac{\log m}{\beta^{2+o(1)}}\right),
\]
which is almost optimal up to the $o(1)$ term. Cohen and Cohen~\cite{CC26} subsequently improved the subpolynomial factor in Ta-Shma's bound. Either construction suffices here.

Let $V\in\bbR^{m\times r}$ have $i$th row $v_i^{\sfT}$. For $g\sim\calN(0,I_r)$, set $G = Vg$. Thus $G\sim\calN(0,VV^{\sfT})$, and
\[
    \Var(G_i)=\|v_i\|_2^2=1,\qquad
    \operatorname{Cov}(G_i,G_j)=\langle v_i,v_j\rangle.
\]
This means that we can sample $G$ by sampling only $r = O(\log m/\beta^{2+o(1)})$ Gaussians instead of $m$. Furthermore, \Cref{cor:small_beta} applies to $G = Vg$.

\paragraph{Discretizing the Distribution.}
We use the following consequence of Kane's discrete Gaussian sampler~\cite[Lemma~3.2]{Kane15}; see also~\cite[Lemma~11]{CHLT19}. For every $0<\eta<1/2$, there is an explicit sampler using $O(\log(1/\eta))$ random bits whose output $\widetilde g_j$ can be coupled with a standard Gaussian $g_j$ so that
\[
    \Pr[|\widetilde g_j-g_j|>\eta]\leq\eta.
\]
Take $r$ independent copies of this coupling and write $\widetilde g=(\widetilde g_1,\ldots,\widetilde g_r)$. A union bound gives
\[
    \Pr[\|g-\widetilde g\|_\infty>\eta]\leq r\eta.
\]
On the complementary event, 
$\|Vg-V\widetilde g\|_\infty
       \leq\|g-\widetilde g\|_2\leq\sqrt r\,\eta.$
Therefore, a sign can change only if $|(Vg)_i|\leq\sqrt r\,\eta$ for some $i$. Each $(Vg)_i$ is a standard Gaussian, so Gaussian anti-concentration and another union bound yield
\begin{equation}\label{eq:prg-sign-error}
    \Pr[\sgn(Vg)\neq\sgn(V\widetilde g)]
       \leq r\eta+\sqrt{\frac{2}{\pi}}\,m\sqrt r\,\eta
       \leq(r+m\sqrt r)\eta.
\end{equation}
For any target error $\eps>0$, we pick $\eta \le \frac{\eps}{2(r+m\sqrt{r})}$.
To sample $\tilde{g}$ we need $s = O(r\cdot \log(1/\eta))$ bits.
Thus, we get the following result as a corollary:
\begin{corollary}\label{cor:prg}
For any $d,L, m \in \bbN, \eps> 0$, there exists an explicit pseudorandom generator $P:\{-1,1\}^s \to\{-1,1\}^m$ that $\eps$-fools all $(d, L)$-AM circuits $C$ with seed-length $s = \poly(d,L,\log(m),1/\eps)$.
\end{corollary}
\begin{proof}
Note that if either $d, L, 1/\eps>m$, then the trivial PRG with seed-length $s = m$ already satisfies $s = \poly(d,L,\log(m), 1/\eps)$. Thus, we can assume without loss of generality that $d, L, 1/\eps\le m$.

Pick $\beta = (\eps/O(dL^2\log^{3/2}(mL)))^4$ so that the error in \Cref{cor:small_beta} is at most $\eps/2$. This is justified since such $\beta$ can control the first error term be at most $\eps/4$ and the second error term is $\frac{1}{(2mL)^4} \le \eps/4$. 

Applying the above construction from~\cite{CHLT19,TaShma17,Kane15} with $r = \log(m)/\beta^{2+o(1)}$ and $\eta \le \eps/(2(r+m\sqrt{r})) = 1/\poly(m)$ gives a seed length of
    \[ s = O(r \log(1/\eta)) = O(\log^2 m) /\beta^{2+o(1)} = \poly(d,L,\log(m),1/\eps). \]
Denote the sampled discrete Gaussian vector by $\tilde{g}$. The output of the PRG is $\sgn(V \tilde{g})$. By the triangle inequality, \Cref{cor:small_beta}, and $\Pr[\sgn(Vg)\neq \sgn(V\tilde{g})]\le \eps/2$, we get
\begin{align*}
    |\E[C(\sgn(V\tilde{g}))] - \E[C(\calU_m)]| &\leq \Pr[\sgn(V\tilde{g})\neq \sgn(Vg)] + |\E[C(\sgn(Vg))] - \E[C(\calU_m)]| \\
    &\leq \eps/2 + \eps/2 = \eps.
\end{align*}
for all $(d,L)$-AM circuits $C$  as desired.
\end{proof}

%% file: bqp_algorithm.tex
\section{\texorpdfstring{$\BQP$}{BQP} Algorithm for Forrelation}
\begin{proof}[Proof of Lemma \ref{lemma:bqp_alg}]
We first establish the following basic facts about the distributions $\calU_n$ and $\calD_n$.
\begin{lemma}\label{lemma:forr_conc_high}
$\Pr_{(x, y) \sim \calD_n}[\forr(x, y) \geq \frac{1}{4}] \geq 1-O\left(\frac{1}{N}\right)$.
\end{lemma}
\begin{proof}[Proof of Lemma \ref{lemma:forr_conc_high}]
Define $\mu := \sqrt{2/\pi}$. Since
\begin{align*}
    \|H_N \sgn(x) - \sgn(y) \|_2^2 &= \|H_N \sgn(x)\|_2^2 +\|\sgn(y) \|_2^2  - 2(H_N \sgn(x))^{\sfT} \sgn(y) \\
    &= 2N - 2(H_N \sgn(x))^{\sfT} \sgn(y) \\
    &= 2N(1-\forr(\sgn(x), \sgn(y))),
\end{align*}
we can rewrite $\forr(\sgn(x), \sgn(y))$ as 
    \[ \forr(\sgn(x), \sgn(y)) = 1-\frac{\|H_N \sgn(x) - \sgn(y)\|_2^2}{2N}. \]
Because $y = H_N x$, we can write
    \[ H_N \sgn(x) - \sgn(y) = H_N (\sgn(x) - \mu x) - (\sgn(y) - \mu y). \]
Applying the fact that $\|u-v\|_2^2 \leq 2\|u\|_2^2 + 2\|v\|_2^2$, we have that
\begin{align*}
    \forr(\sgn(x), \sgn(y)) &= 1-\frac{\|H_N \sgn(x) - \sgn(y)\|_2^2}{2N} \\
    &\geq 1-\frac{\|H_N (\sgn(x) - \mu x)\|_2^2 + \|\sgn(y) - \mu y\|_2^2}{N} \\
    &= 1-\frac{\|\sgn(x) - \mu x\|_2^2 + \|\sgn(y) - \mu y\|_2^2}{N}.
\end{align*}
Now define $A := \frac{\|\sgn(x) - \mu x\|_2^2}{N}$ and $B := \frac{\|\sgn(y) - \mu y\|_2^2}{N}$ (note that the marginals of $x$ and $y$ and hence $A$ and $B$ are identical). As
\begin{align*}
    A = \frac{1}{N} \sum_{i = 1}^N (\sgn(x_i)-\mu x_i)^2 = \frac{1}{N} \sum_{i = 1}^N (1-\mu |x_i|)^2,
\end{align*}
$A$ and $B$ are each the average of $N$ independent random variables $W_i := (1-\mu|Z_i|)^2$ where $Z_i \sim \calN(0, 1)$. Basic Gaussian moments give 
    \[ \bbE[W_i] = 1-2\mu\bbE[|Z_i|]+\mu^2\bbE[Z_i^2] = 1-2\mu^2+\mu^2 = 1-\frac{2}{\pi}, \]
and 
    \[ \bbE[W_i^2] = \bbE[(1-\mu|Z_i|)^4] = \bbE[(|1-\mu|Z_i||)^4] \leq \bbE[(\max\{1, |Z_i|\})^4] \leq 1+\bbE[Z_i^4] = 4. \]
Thus
    \[ \bbE[A+B] = 2-\frac{4}{\pi}, \qquad \Var(A+B) \leq 2 \cdot (\Var(A)+\Var(B)) \leq \frac{16}{N}. \]
By Chebyshev's inequality we conclude that
\begin{align*}
    \Pr_{(x, y) \sim \calD_n}[\forr(x, y) < 1/4] &= \Pr_{x \sim \calN(0, I_N), y = H_N x}[\forr(\sgn(x), \sgn(y)) < 1/4] \\
    &\leq \Pr_{x \sim \calN(0, I_N), y = H_N x}[A + B > 3/4] \\
    &\leq \frac{\Var(A+B)}{(3/4-\bbE[A+B])^2} \\
    &\leq \frac{16}{N(1/50)^2} = \frac{40000}{N},
\end{align*}
as desired.
\end{proof}

\begin{lemma}\label{lemma:forr_conc_low}
$\Pr_{(x, y) \sim \calU_n}[\forr(x, y) \leq \frac{1}{16}] \geq 1-O\left(\frac{1}{N}\right)$.
\end{lemma}
\begin{proof}[Proof of Lemma \ref{lemma:forr_conc_low}]
The calculation is fairly simple: independence of each coordinate gives
    \[ \bbE_{(x, y) \sim \calU_n}\left[\forr(x, y)^2\right] = \frac{1}{2^{3n}} \sum_{i, j} 1 = \frac{1}{N}. \]
Applying Markov's inequality, we conclude that $\Pr_{(x, y) \sim \calU_n}[\forr(x, y) > \frac{1}{16}] \leq \frac{256}{N}$.
\end{proof}

With Lemmas \ref{lemma:forr_conc_high} and \ref{lemma:forr_conc_low} in hand, we can now appeal to the 1-query algorithm for Forrelation due to Aaronson and Ambainis \cite{AA15}:
\begin{lemma}[Modified from Proposition 6, \cite{AA15}]\label{lemma:forr_alg}
There exists a 1-query, $O(n)$-gate quantum algorithm $\calA$ such that for any functions $f, g \in \{\pm 1\}^N$, 
    \[ \Pr[\calA^{f, g}(1^n) = 1] = \frac{1+\forr(f, g)}{2}. \]
\end{lemma}

By repeating the algorithm from Lemma \ref{lemma:forr_alg} sufficiently many times and accepting if at least $9/16$ of runs accept, we have by a Chernoff bound that there exists an $O(1)$-query and $O(n)$-gate quantum algorithm $\calB$ such that
    \[ \Pr_{(x, y) \sim \calD_n}\left[\Pr_{\calB}[\calB^{x, y} = 1] \geq \frac{2}{3}\right] \geq \Pr_{(x, y) \sim \calD_n}\left[\forr(x, y) \geq \frac{1}{4}\right] \geq 1-O\left(\frac{1}{N}\right), \]
and
    \[ \Pr_{(x, y) \sim \calU_n}\left[\Pr_{\calB}[\calB^{x, y} = 1] \leq \frac{1}{3}\right] \geq \Pr_{(x, y) \sim \calU_n}\left[\forr(x, y) \leq \frac{1}{16}\right] \geq 1-O\left(\frac{1}{N}\right), \]
as desired.
\end{proof}

%% file: general_smoothing.tex
\section{Smooth Extensions for General Boolean Functions} 
\begin{lemma}
\label{lemma:arbitrary-smooth-extension}
There exists a constant $C$ such that, for every $n \geq 1$ and every function $f: \{-1,1\}^n \to [-1, 1]$, there is a smooth function $F: \bbR^n \to [-1, 1]$ satisfying $F(z) = f(z)$ for all $z \in \{-1, 1\}^n$ and
\[
    \sup_{x\in\mathbb{R}^n}\|\nabla F(x)\|_1 \leq 2,
    \qquad
    \sup_{x\in\mathbb{R}^n}\|\nabla^2 F(x)\|_1 \leq C\log(2n).
\]
\end{lemma}
\begin{proof}
For each vertex $z \in \{\pm 1\}^n$, we construct a smooth function $\psi_z$ that, together with all its derivatives, vanishes at any point $x$ unless every coordinate $x_i$ has sign $z_i$. These regions are disjoint for $\psi_z$ and $\psi_{z'}$ if $z \neq z'$. Therefore, at each point, the derivatives of $F = \sum_z f(z)\psi_z$ receive a contribution from at most one summand. Since $|f(z)| \leq 1$, any uniform derivative bounds for the individual functions $\psi_z$ also hold for $F$.

First, fix a smooth nondecreasing function $\rho:\mathbb R\to[0,1]$ satisfying all of the following properties:
    \[ \rho(t) = 0 \quad(t \leq 0), \qquad \rho(t) = 1 \quad (t \geq 3/4), \qquad \|\rho'\|_\infty \leq 2. \]
For example, choose a smooth function $h: \bbR \to [0, 1]$ that vanishes outside $(0, 3/4)$ and satisfies $\int_{\bbR} h > 1/2$, and set
    \[ \rho(t) = \frac{\int_{-\infty}^t h(s)\,ds}{\int_{\bbR} h(s)\,ds}. \]
Fixing this choice also makes $\|\rho''\|_\infty$ an absolute constant.

Set $\lambda = 4\log(2n)$. For each $z \in \{\pm 1\}^n$, define the functions
\[
q_z(x) = -\frac{1}{\lambda} \log\!\left(\sum_{i = 1}^n e^{-\lambda z_ix_i}\right), \qquad \psi_z(x) = \rho(q_z(x)).
\]
Here $q_z$ is a smooth approximation to $\min_i z_ix_i$, with $q_z(x) \leq \min_i z_ix_i$ and $q_z(z) = 1-\frac{\log n}{\lambda} \geq 3/4$. Thus $\psi_z(z) = 1$. Moreover, for every $x \in \bbR^n$,
\[
    \psi_z(x) \neq 0
    \implies q_z(x) > 0
    \implies z_i x_i > 0 \quad \forall i \in [n].
\]
In particular, $\psi_z$ vanishes at every Boolean vertex other than $z$.

To bound the derivatives, we refer to a work by Chernozhukov et al. which shows that the smoothmax function $S_\lambda(u)$ satisfies $\|\nabla S_\lambda(u)\|_1 = 1$ and $\|\nabla^2S_\lambda(u)\|_1 \leq 2\lambda$~\cite[Lemma A.2]{CCK13}. Since $q_z(x) = -S_\lambda(-z_1x_1,\ldots,-z_nx_n)$, the same norm bounds hold for $q_z$. The chain rule therefore gives
    \[ \|\nabla\psi_z(x)\|_1 = |\rho'(q_z(x))| \cdot \|\nabla q_z(x)\|_1 \leq 2\cdot 1 = 2, \]
and
\begin{align*}
    \|\nabla^2\psi_z(x)\|_1 &\leq \|\rho''\|_\infty \|\nabla q_z(x)\|_1^2 + \|\rho'\|_\infty \|\nabla^2q_z(x)\|_1 \\
    &\leq \|\rho''\|_\infty + 4\lambda = O(\log(2n)).
\end{align*}

Finally, set $F(x) = \sum_{z \in \{\pm 1\}^n} f(z)\psi_z(x)$. At any point $x$, at most one $q_z(x)$ is positive. Whenever $q_z(x) \leq 0$, all derivatives of $\rho$ vanish at $q_z(x)$. By the chain rule, $\psi_z(x) = 0$ and all derivatives of $\psi_z$ vanish at $x$ as well. Hence $F(x) \in [-1,1]$, and the derivative bounds above extend directly to $F$. The vertex identities give $F(z) = f(z)$, and $F$ is smooth because the sum is finite.
\end{proof}

%% file: diagonalization.tex
\section{Diagonalization Argument}
\begin{proof}[Proof of Corollary \ref{cor:separation}]
    The diagonalization proof is standard, but we include it for completeness, following \cite{Aar10, RT22}. First, it is easy to see that the Goldwasser-Sipser transformation \cite{GS86} is classically relativizing, that is, $\IP^{\calO} = \AM[\poly]^{\calO}$ for all classical oracles $\calO$. It thus suffices to find an oracle $\calO$ such that $\BQP^{\calO} \not\subseteq \AM[\poly]^{\calO}$.
    
    Let $\calL$ be a random unary language. The oracle $\calO$ encodes the truth tables of Boolean functions of different input lengths. For each $n$, if $1^n \in \calL$ then we draw $(x_n, y_n) \in \{\pm 1\}^{2N}$ from the Forrelated distribution $\calD_n$, and if $1^n \notin \calL$ then we draw $(x_n, y_n)$ from the uniform distribution $\calU_n$. We interpret $(x_n, y_n) \in \{\pm 1\}^{2N}$ as a Boolean function $f_n : \{0, 1\}^{n+1} \to \{0, 1\}$ that describes the truth table of $\calO$ restricted to queries of length $n+1$.\footnote{We can do this by taking the canonical bijection between $[2N]$ and $\{0, 1\}^{n+1}$ and mapping signs $1$ and $-1$ to outputs $0$ and $1$, respectively.}
    
    First, note that there exists a $\BQP$ machine $M$ that correctly decides $\calL$ with high probability over the choice of $\calL$ and $\calO$. The machine $M$, on input $1^n$, runs the quantum algorithm $\calA$ from Lemma \ref{lemma:bqp_alg} on the oracle string provided by $\calO$ of length $n+1$. Note that this is a $\BQP$ machine since $\calA$ runs in $\poly(n)$ time. In addition, assuming $n$ is sufficiently large,
    \begin{itemize}
        \item If $1^n \in \calL$, then $(x_n, y_n)$ was sampled from $\calD_n$, and thus
            \[ \Pr_{\calO \sim \calD_n}\left[\Pr_M[M^{\calO}(1^n) = 1] \geq \frac{2}{3}\right] \geq 1-O\left(\frac{1}{N}\right) \geq 1-\frac{1}{100n^2}. \]
        \item If $1^n \notin \calL$, then $(x_n, y_n)$ was sampled from $\calU_n$, and thus
            \[ \Pr_{\calO \sim \calU_n}\left[\Pr_M[M^{\calO}(1^n) = 1] \leq \frac{1}{3}\right] \geq 1-O\left(\frac{1}{N}\right) \geq 1-\frac{1}{100n^2}. \]
    \end{itemize}
    We see that in both cases the probability over $\calO$ that $M$ decides $\calL$ correctly on $1^n$ is at least $1-\frac{1}{100n^2}$ assuming $n \geq n_0$ for some constant $n_0$. Thus,
        \[ \Pr_{\calL, \calO}[\text{$M^{\calO}$ decides $\calL$ correctly on all $n \geq n_0$}] \geq \prod_{n = n_0}^{\infty} \left(1-\frac{1}{100n^2}\right) \geq 0.9. \]

    We now argue that $\calL \notin \AM[\poly]^{\calO}$ with high probability over the choice of $\calL$ and $\calO$. Begin by fixing an $\AM[\poly]$ verifier $\calB$. For a fixed choice of oracle $\calO$ and $\calL$, let $E_n(\calB, \calL, \calO)$ be the event that $\calB^{\calO}$ decides $\calL$ correctly on $1^n$. This means that if $1^n \in \calL$, then
        \[ \max_P \Pr_{\calB}[\langle \calB^{\calO}(1^n) \leftrightarrow P^{\calO} \rangle = 1] \geq \frac{2}{3}, \]
    while if $1^n \notin \calL$, then 
        \[ \max_P \Pr_{\calB}[\langle \calB^{\calO}(1^n) \leftrightarrow P^{\calO} \rangle = 1] \leq \frac{1}{3}. \]
    
    By Lemma \ref{lemma:am_characterization}, the acceptance probability function $a_{\calB}: \{\pm 1\}^{2N} \to [0, 1]$ of $\calB$ can be computed by a $(\poly(n), \poly(n))$-AM circuit $C_{\calB}$. Corollary \ref{cor:values} then implies that for sufficiently large $n$,
        \[ \bigl|\bbE_{(x, y) \sim \calD_n}[a_{\calB}(x, y)] - \bbE_{(x, y) \sim \calU_n}[a_{\calB}(x, y)]\bigr| = \bigl|\bbE_{(x, y) \sim \calD_n}[C_{\calB}(x, y)] - \bbE_{(x, y) \sim \calU_n}[C_{\calB}(x, y)]\bigr| \leq \frac{1}{15}. \]
    However,
        \[ \bbE_{(x, y) \sim \calD_n}[a_{\calB}(x, y)] \geq \frac{2}{3} \cdot \Pr_{(x, y) \sim \calD_n}\left[a_{\calB}(x, y) \geq \frac{2}{3}\right], \]
    and 
    \begin{align*}
        \bbE_{(x, y) \sim \calU_n}[a_{\calB}(x, y)] &\leq \frac{1}{3} \cdot \Pr_{(x, y) \sim \calU_n}\left[a_{\calB}(x, y) \leq \frac{1}{3}\right] + 1 \cdot \left(1 - \Pr_{(x, y) \sim \calU_n}\left[a_{\calB}(x, y) \leq \frac{1}{3}\right]\right) \\
        &= 1 - \frac{2}{3} \cdot \Pr_{(x, y) \sim \calU_n}\left[a_{\calB}(x, y) \leq \frac{1}{3}\right],
    \end{align*}
    so
    \begin{align*}
        \Pr_{\calL, \calO}[E_n(\calB, \calL, \calO)] &= \frac{1}{2} \Pr_{(x, y) \sim \calD_n}\left[a_{\calB}(x, y) \geq \frac{2}{3}\right] + \frac{1}{2} \Pr_{(x, y) \sim \calU_n}\left[a_{\calB}(x, y) \leq \frac{1}{3}\right] \\
        &\leq \frac{3}{4} \cdot \bbE_{(x, y) \sim \calD_n}[a_{\calB}(x, y)] + \frac{3}{4} - \frac{3}{4} \cdot \bbE_{(x, y) \sim \calU_n}[a_{\calB}(x, y)] \\ 
        &\leq \frac{3}{4} + \frac{3}{4} \cdot \frac{1}{15} = \frac{4}{5}.
    \end{align*}
    By independence of $\calO$ on different input lengths, and the fact that $\calB$ can only ask queries of length $\poly(n)$ on input $1^n$, we get that there are infinitely many input lengths $n_1, n_2, \ldots$ such that for each $i \in \bbN$,
        \[ \Pr_{\calL, \calO}[E_{n_{i+1}}(\calB, \calL, \calO) \mid E_{n_1}(\calB, \calL, \calO) \land \ldots \land E_{n_i}(\calB, \calL, \calO)] \leq \frac{4}{5}. \]
    Thus we conclude that 
        \[ \Pr_{\calL, \calO}[E_1(\calB, \calL, \calO) \land E_2(\calB, \calL, \calO) \land \ldots] = 0, \]
    and since there are countably many $\AM[\poly]$ verifiers,
        \[ \Pr_{\calL, \calO}\left[\calL \in \AM[\poly]^{\calO}\right] =  \Pr_{\calL, \calO}[\exists \calB: E_1(\calB, \calL, \calO) \land E_2(\calB, \calL, \calO) \land \ldots] = 0. \]
    Therefore, there exists an oracle $\calO$ and language $\calL$ such that $M^{\calO}$ decides $\calL$ correctly on all $n \geq n_0$ but no $\AM[\poly]$ verifier decides $\calL$ correctly. Hardwiring the membership of $1^n$ for all $n < n_0$ into $M$ then gives us our desired separation.
\end{proof}

%% file: bbbv_cor.tex
\section{Oracle Separation Between \texorpdfstring{$\BQP \cap \NP$}{BQP cap NP} and \texorpdfstring{$\IP_{\BQP}$}{IP BQP}}

\begin{proof}[Proof of Corollary \ref{cor:simons_like}]
    Let $\calL$ be a random unary language. The oracle $\calO$ consists of two oracles $\calO_1$ and $\calO_2$, each of which encodes the truth tables of Boolean functions of increasing input lengths.
    
    We first describe how to sample $\calO_1$: For each $n$, if $1^n \in \calL$ then we draw $x_n \in \{\pm 1\}^{2N}$ from the Forrelated distribution $\calD_n$, and if $1^n \notin \calL$ then we draw $x_n$ from the uniform distribution $\calU_n$. We interpret $x_n \in \{\pm 1\}^{2N}$ as a Boolean function $f_n : \{0, 1\}^{n+1} \to \{0, 1\}$ that describes the oracle $\calO_1$ restricted to strings of length $n+1$.

    We now describe how to sample $\calO_2$: For each $n$, if $1^n \in \calL$ then we draw $x'_n \in \{\pm 1\}^{2N}$ from the uniform distribution over strings with exactly one $-1$, and if $1^n \notin \calL$ then we set $x'_n = (+1)^{2N}$. We interpret $x'_n \in \{\pm 1\}^{2N}$ as a Boolean function $f'_n : \{0, 1\}^{n+1} \to \{0, 1\}$ that describes the oracle $\calO_2$ restricted to strings of length $n+1$.

    As in the proof of Corollary \ref{cor:separation}, it follows that there exists some constant $n_0$ and $\BQP$ machine $M$ such that 
        \[ \Pr_{\calL, \calO}[\text{$M^{\calO}$ decides $\calL$ correctly on all $n \geq n_0$}] \geq 0.9. \]
    This is since $M$ can simply ignore $\calO_2$ and the situation is reduced to that of Corollary \ref{cor:separation}.

    It is also not hard to see that $\calL$ is always in $\NP^{\calO}$ by considering the following $\NP$ machine $M'$: for each input $1^n$, given a witness $w \in [2N]$ (which is of length $\log(2N) = \poly(n)$), query $\calO_2$ at the $w$th index and accept iff $\calO_2(w) = -1$. It is straightforward to see that if $1^n \in \calL$ then for any choice of $\calO_2$ which is sampled there exists a witness $w$ that makes $M'$ accept, while if $1^n \notin \calL$ then no witness makes $M'$ accept.

    We now argue that $\calL \notin \IP_{\BQP}^{\calO}$ with high probability over the choice of $\calL$ and $\calO$. Begin by fixing an $\IP_{\BQP}$ verifier $\calA$ (which also fixes an honest $\BQP$ prover $P$). By straightforward amplification (which is relativizing) we can assume without loss of generality that the proof system has soundness and completeness error at most $\frac{1}{10}$. For a fixed choice of oracle $\calO$ and $\calL$, let $F_n(\calA, \calL, \calO)$ be the event that $\calA^{\calO}$ decides $\calL$ correctly on $1^n$. This means that if $1^n \in \calL$, then
        \[ \Pr_{\calA}[\langle \calA^{\calO_1, \calO_2}(1^n) \leftrightarrow P^{\calO_1, \calO_2} \rangle = 1] \geq \frac{9}{10}, \]
    while if $1^n \notin \calL$, then 
        \[ \max_{P^{*}} \Pr_{\calA}[\langle \calA^{\calO_1, \calO_2}(1^n) \leftrightarrow (P^{*})^{\calO_1, \calO_2} \rangle = 1] \leq \frac{1}{10}. \]
    We now note that a standard hybrid argument due to \cite{BBBV97} implies that because $\calA$ and $P$ are $\BPP$ and $\BQP$ machines which make at most $\poly(n)$ queries to $\calO_2$, they can distinguish between when $\calO_2$ is drawn from YES and NO distributions with advantage at most $\frac{\poly(n)}{2N}$. We can therefore conclude that for sufficiently large $n$, if $1^n \in \calL$, then
        \[ \Pr_{\calA, \calO_1 \sim \calD_n, \calO_2 = -1^{2N}}[\langle \calA^{\calO_1, \calO_2}(1^n) \leftrightarrow P^{\calO_1, \calO_2} \rangle = 1] \geq \frac{9}{10}-\frac{1}{10} \geq \frac{2}{3}, \]
    while if $1^n \notin \calL$, then 
        \[ \max_{P^{*}} \Pr_{\calA, \calO_1 \sim \calU_n, \calO_2 = -1^{2N}}[\langle \calA^{\calO_1, \calO_2}(1^n) \leftrightarrow (P^{*})^{\calO_1, \calO_2} \rangle = 1] \leq \frac{1}{10} \leq \frac{1}{3}. \]
    A straightforward simulation argument implies that $\langle \calA, P \rangle$ is an $\IP_{\BQP}$ protocol for distinguishing between $\calO_1 \sim \calD_n$ and $\calO_1 \sim \calU_n$! Applying the Goldwasser-Sipser transformation \cite{GS86} here now gives an $\AM[\poly]$ verifier which distinguishes between $\calO_1 \sim \calD_n$ and $\calO_1 \sim \calU_n$. In other words, for sufficiently large $n$, the event $F_n(\calA, \calL, \calO)$ implies the event $E_n(\calB, \calL, \calO_1)$ defined in Corollary \ref{cor:separation} for some $\AM[\poly]$ verifier $\calB$. Thus, the proof of Corollary \ref{cor:separation} implies that for sufficiently large $n$,
        \[ \Pr_{\calL, \calO}[F_n(\calA, \calL, \calO)] \leq \Pr_{\calL, \calO_1}[E_n(\calB, \calL, \calO_1)] \leq \frac{4}{5}. \]
    We therefore conclude that
        \[ \Pr_{\calL, \calO}[F_1(\calA, \calL, \calO) \land F_2(\calA, \calL, \calO) \land \ldots] = 0, \]
    and since there are countably many $\IP_{\BQP}$ verifiers,
        \[ \Pr_{\calL, \calO}\left[\calL \in \IP_{\BQP}^{\calO}\right] =  \Pr_{\calL, \calO}[\exists \calA: F_1(\calA, \calL, \calO) \land F_2(\calA, \calL, \calO) \land \ldots] = 0. \]
    Therefore, there exists an oracle $\calO$ and language $\calL$ such that $M^{\calO}$ decides $\calL$ correctly on all $n \geq n_0$ but no $\IP_{\BQP}$ verifier decides $\calL$ correctly. Hardwiring the membership of $1^n$ for all $n < n_0$ into $M$ places $\calL \in \BQP^{\calO}$, while our previous observation shows that $\calL \in \NP^{\calO}$, which finishes the separation.
\end{proof}